\documentclass[vecarrow,envcountresetchap,envcountsame,envcountsect]{svmult}
\usepackage{mathptmx}        
\usepackage[bottom]{footmisc}
\usepackage[pdfborder={0 0 0}]{hyperref}
\usepackage{amssymb}
\usepackage{amsmath}
\usepackage{stmaryrd}
\usepackage{diagrams}

\newcommand{\dsem}[1]{[\![{#1}]\!]}
\newcommand{\lsem}{[\![}
\newcommand{\rsem}{]\!]}

\newcommand{\comp}[2]{#1\raisebox{.3ex}{\tiny o}#2}  
\newcommand{\type}{\!:\!}

\newcommand{\mystop}{\mathtt{Stop}}

\newcommand{\mypar}{\mid\mid} 
\newcommand{\myinterleave}{\mid\mid\mid} 
\newcommand{\prog}[1]{\mathtt{#1}\,}
\newcommand{\Th}{\mathrm{Th}}
\newcommand{\myaction}{\rightarrow}
\newcommand{\myactionsub}{\rightarrow}
\newcommand{\boxaction}{}

\newcommand{\trF}{\mathrm{tr}_{F}}
\newcommand{\fut}{\mathrm{fut}}
\newcommand{\futF}{\mathrm{fut}_{F}}
\newcommand{\futP}{\mathrm{fut}_{P}}

\newcommand{{\finitary}}{finitary}

\newcommand{\T}{{\cal T}}
\newcommand{\R}{{\cal R}}
\newcommand{\F}{{\cal F}}
\newcommand{\fR}{{{\cal R}\!_f}}
\newcommand{\fF}{{{\cal F}\!\!\!_f}}
\newcommand{\dT}{{{\cal T}_d}}
\newcommand{\dR}{{{\cal R}_d}}
\newcommand{\dF}{{{\cal F}\!\!_d}}

\newcommand{\dfR}{{{\cal R}_{\it df}}}
\newcommand{\dfF}{{{\cal F}\!\!_{\it df}}}
\newcommand{\TX}{{{\cal T}(X)}}
\newcommand{\RX}{{{\cal R}(X)}}
\newcommand{\FX}{{{\cal F}(X)}}
\newcommand{\TXY}{{{\cal T}(X,Y)}}
\newcommand{\RXY}{{{\cal R}(X,Y)}}

\newcommand{\dTX}{{{\cal T}_d(X)}}
\newcommand{\dRX}{{{\cal R}_d(X)}}
\newcommand{\dFX}{{{\cal F}\!\!_d(X)}}

\newcommand{\dfFX}{{{\cal F}\!\!_{\it df}(X)}}
\newcommand{\dfTXY}{{{\cal T}_{\it df}(X,Y)}}
\newcommand{\dfRXY}{{{\cal R}_{\it df}(X,Y)}}

\begin{document}

\title*{On CSP and the Algebraic Theory of Effects}
\author{Rob van Glabbeek and Gordon Plotkin\thanks{This work was done
    with the support of a Royal Society-Wolfson Award.}}
\institute{Rob van Glabbeek \at NICTA, Sydney, Australia\\
  University of New South Wales, Sydney, Australia
\and Gordon Plotkin \at Laboratory for the Foundations of Computer Science,
     School of Informatics, University of Edinburgh, UK}
\maketitle
\thispagestyle{empty}

\abstract*{We consider CSP from the point of view of the algebraic
theory of effects, which classifies operations as effect
\emph{constructors} or effect \emph{deconstructors}; it also
provides a link with functional programming, being a refinement of
Moggi's seminal monadic point of view. There is a natural algebraic 
theory of the constructors whose free algebra functor
is Moggi's monad;  we illustrate this by characterising free and
initial algebras in terms of two versions of the stable failures
model of CSP, one more general than the other. Deconstructors are
dealt with as homomorphisms to  (possibly non-free) algebras.
   One can view CSP's action and choice operators as
constructors and the rest, such as  concealment and concurrency, as
deconstructors. Carrying this programme out results in taking
deterministic external choice as constructor rather than general
external choice. However, binary deconstructors, such as the CSP
concurrency operator, provide unresolved difficulties.  We conclude
by presenting a combination of CSP with Moggi's computational
$\lambda$-calculus, in which the operators, including concurrency,
are polymorphic. While the paper mainly concerns CSP, it ought to be
possible to carry over similar ideas to other process calculi.
}

\abstract{We consider CSP from the point of view of the algebraic
theory of effects, which classifies operations as effect
\emph{constructors} and effect \emph{deconstructors}; it also
provides a link with functional programming, being a refinement of
Moggi's seminal monadic point of view. There is a natural algebraic 
theory of the constructors whose free algebra functor
is Moggi's monad;  we illustrate this by characterising free and
initial algebras in terms of two versions of the stable failures
model of CSP, one more general than the other. Deconstructors are
dealt with as homomorphisms to  (possibly non-free) algebras.
\newline\indent
One can view CSP's action and choice operators as
constructors and the rest, such as  concealment and concurrency, as
deconstructors. Carrying this programme out results in taking
deterministic external choice as constructor rather than general
external choice. However, binary deconstructors, such as the CSP
concurrency operator, provide unresolved difficulties.  We conclude
by presenting a combination of CSP with Moggi's computational
$\lambda$-calculus, in which the operators, including concurrency,
are polymorphic. While the paper mainly concerns CSP, it ought to be
possible to carry over similar ideas to other process calculi.
}

\section{Introduction}

We examine Hoare's CSP~\cite{BHR84,Hoa85,Ros98} from the point of view of the algebraic theory of effects~\cite{PP02,PP04,HPP06,PPr09},  a refinement of Moggi's  seminal `monads as notions of computation'~\cite{Mog89,Mog91,BHM02}. This is a natural exercise as the  algebraic nature of both points to a possibility of commonality. 
In the algebraic theory of effects operations do not all have the same character. Some are effect \emph{constructors}: they create the effects at hand; some are effect \emph{deconstructors}: they respond to effects created. For example, raising an exception creates an effect---the exception raised---whereas exception-handling responds to effects---exceptions that have been raised. It may therefore be interesting, and even useful,  to classify CSP operators as constructors or deconstructors. Considering CSP and the algebraic theory of effects together also raises the possibility of combining CSP with functional programming in a principled way, as Moggi's monadic approach provides a framework for the combination of computational effects with functional programming. More generally, although we mainly consider CSP, a similar exercise could be undertaken for other process calculi as they have a broadly similar algebraic character.

The theory of algebraic effects starts with the observation that effect constructors generally satisfy natural equations, and Moggi's monad $T$ is precisely
the free algebra monad for these equations (an exception is the continuations monad which is of a different character). 
Effect deconstructors are treated as homomorphisms from the free algebra to another algebra, perhaps with the same carrier as the free algebra but with different operations. These operations can be given by combinations of effect constructors and previously defined deconstructors. The situation is much like that of primitive recursive definitions, although we will not present a formal definitional scheme.

We mainly consider that part of CSP containing action, internal and
external choice, deadlock, relabelling, concealment, concurrency and
interleaving, but not, for example,
recursion (we do, albeit briefly, consider the extension with termination and sequencing). The evident constructors are then action prefix, and the two kinds of choice, internal and external, the latter together with deadlock. The evident deconstructors are relabelling, concealment, concurrency and interleaving.  There is, however, a fly in the ointment, as pointed out in~\cite{PPr09}.  Parallel operators, such as CSP's concurrency and interleaving,  are naturally binary, and respond to effects in both arguments. However, the homomorphic approach to deconstructors, as sketched above, applies only to unary deconstructors, although it is possible to extend it to accommodate parameters and simultaneous definitions. Nonetheless, the natural definitions of concurrency and interleaving do not fall within the homomorphic approach, even in the extended sense. This problem has
nothing to do with CSP: it  applies to all examples of parallelism of which we are aware.

Even worse, when we try to carry out the above analysis for CSP, it seems that the homomorphic approach cannot handle concealment. The difficulty is caused by the fact that concealment does not commute with external choice. Fortunately this difficulty can be overcome by changing the effect constructors: we remove external choice and action prefix and replace them by the deterministic external choice operator 
$(a_{1} \rightarrow P(a_{1}) \mid \ldots \mid a_{n} \rightarrow P(a_{n}))$, where
the $a_{i}$ are all different. Binary external choice then becomes a deconstructor.

 With that we can carry out the program of analysis, finding only the expected difficulty in dealing with concurrency and interleaving. However, it must be admitted that the $n$-ary operators are somewhat clumsy to work with, and it is at least a priori odd to take binary external choice as a deconstructor. On the other hand, in~\cite[Section 1.1.3]{Hoa85} Hoare writes:
 \begin{quote}
 The definition of choice can readily be extended to more than two alternatives, e.g.,
                   \[ (x \rightarrow P \mid y \rightarrow Q \mid \ldots \mid z \rightarrow R)\]
    Note that the choice symbol
    $\mid$ is \emph{not} an operator on processes; it
    would be syntactically incorrect to write  $P\mid Q$, for processes $P$
    and $Q$. The reason for this rule is that we want to avoid giving
    a meaning to
                           \[(x  \rightarrow P \mid x \rightarrow Q)\]
    which appears to offer a choice of first event, but actually fails to do so.
 \end{quote}
 which might be read as offering some support to a treatment which  takes deterministic external choice as a primitive (here = constructor), rather than general external choice.
  On our side, we count it as a strength of the algebraic theory of effects that it classifies effect-specific operations and places constraints on them: that they either belong to the basic theory or must be defined according to a scheme that admits inductive proofs.

Turning to the combination with functional programming, consider  Moggi's computational $\lambda$-calculus. Just as one accommodates imperative programming\linebreak within functional programming by treating commands as expressions of type $\prog{unit}$, so it is natural to treat our selection of CSP terms as expressions of type $\prog{empty}$ as they do not terminate normally, only in deadlock. For process languages such as ACP~\cite{BK85,BK86} which do have the possibility of normal termination, or CSP  with such a termination construct, one switches to regarding process terms as expressions of type $\prog{unit}$, when a sequencing operator is also available.

As we have constructors for every $T(X)$, it is natural to treat them  as polymorphic constructs, rather than just as process combinators. For example, one could have a binary construction for internal choice, with typing rule:
\[\frac{M \type \sigma \quad\quad  N \type \sigma}{M \sqcap N \type \sigma}\]
It is  natural to continue this theme for the deconstructors, as in:
\[\frac{M \type \sigma}{M\backslash a \type \sigma}
\quad\quad\quad
\frac{M \type \sigma \quad\quad N \type \tau}{M \!\mypar\! N  \type \sigma\times\tau}\]
where the thought behind the last rule is that $M$ and $N$ are evaluated concurrently, terminating normally only if they both do, when the pair of results returned individually by each is returned. 

{In the case of CSP a functional programming language CSPM incorporating CSP processes has been given by Scattergood~\cite{Sca}; it is used by most existing CSP tools including the Failures Divergences Refinement Checker (FDR), see~\cite{Ros94}.}
Scattergood's CPSM differs from our proposal in several
respects. Most significantly, processes are not treated on a par with other expressions: in particular they cannot be taken as arguments in functions, and CSP constructors and deconstructors are only available for processes. It remains to be seen if such differences are of practical relevance.

In Section~\ref{CSPconstructors1} we  take deadlock, action, binary internal and external choice as the constructors. We show, in
Theorem~\ref{Initial1}, that, with the standard equational theory, the
initial algebra is the `{\finitary} part'  of the original
Brookes-Hoare-Roscoe failures model~\cite{BHR84}; which is known to be
isomorphic to the {\finitary}, divergence- and $\checkmark$-free part of the
failures/divergences model, {as well as
the {\finitary}, divergence- and $\checkmark$-free part of the
stable failures model, both of which are described in \cite{Ros98}.}
In Section~\ref{Deconstructors} we go on to consider effect deconstructors, arriving at the difficulty with concealment and illustrating the problems with parallel operators in the (simpler) context of Milner's synchronisation trees. A reader interested in the problem of dealing with parallel operators algebraically need only read this part, together with~\cite{PPr09}.

We then backtrack in Section~\ref{CSPconstructors2},  making a different choice of constructors, as discussed above, and giving another characterisation of the {\finitary} failures model as an initial algebra in Theorem~\ref{Initial2}. With that, we can carry out our programme, failing only where expected: with the binary deconstructors.
In Section~\ref{divergence} we add a zero for the internal choice
operator to our algebra; this can be interpreted as divergence in the
stable failures model, and permits the introduction of a useful
additional deterministic external choice constructor. Armed with this tool,
in Section~\ref{Functional}, we look at the combination of CSP and functional programming, following the lines hinted at above. In order to give a denotational semantics we need, in Theorem~\ref{Free}, to characterise the free algebras rather than just the initial one.  

As remarked above, termination and sequencing are accommodated within functional programming  via the type 
$\prog{unit}$; {in Section~\ref{termination} }we therefore also give a brief treatment of our
fragment of CSP extended with termination and sequencing, modelling it
in the free algebra over the one-point set. 

The concluding Section~\ref{the-end} contains a brief discussion of the general question of combining process calculi, or parallelism with a global store, with functional programming. The case of CSP considered here is just one example of the many possible such combinations. 
Throughout this paper we do not consider recursion;
this enables us to work within the category of sets. A more complete treatment would deal with recursion working within, say, the category of $\omega$-cpos (i.e., partial orders with lubs of increasing $\omega$-sequences) and continuous functions (i.e., monotone functions preserving lubs of increasing $\omega$-sequences). This is discussed further in Section~\ref{the-end}.
The appendix gives a short presentation of Moggi's computational $\lambda$-calculus.

\section{Technical preliminaries} \label{Prelim}

We give a brief sketch of finitary equational theories and their free algebra monads. For a fuller explanation see, e.g.,~\cite{Bor94,AGM95}.
Finitary equational theories $\Th$ are derived from a given set of axioms, written
using a signature $\Sigma$ consisting of a set of operation symbols $\mathrm{op}: n$, together with their arities $n \geq 0$. One forms terms $t$ from the signature and variables and the axioms then consist of equations $t = u$ between the terms; there is a natural equational logic for deducing consequences of the axioms;
and the theory consists of all the equations derivable from the axioms.  A \emph{ground} equation is one where both terms are \emph{closed}, meaning that they contain no variables.

For example, we might consider the fragment of CSP
with signature $\boxempty \type 2$, $\mystop\type0$ and the following axioms for a semilattice (the first three axioms) with a zero (the last):
	\[\begin{array}{lc}
	\mbox{Associativity} & (x \boxempty  y) \boxempty  z = x \boxempty  (y \boxempty  z)\\
	\mbox{Commutativity} & x \boxempty  y = y \boxempty  x\\
	\mbox{Idempotence}    & x \boxempty  x = x\\
	\mbox{Zero}          &  x \boxempty  {\mystop}   = x
	\end{array}\]
A \emph{$\Sigma$-algebra} is a structure $\mathcal{A} = (X, (\mathrm{op}_{\mathcal{A}}\type X^{n} \rightarrow X)_{\mathrm{op}: n \in \Sigma})$; we say that $X$ is the \emph{carrier} of $\mathcal{A}$ and the $\mathrm{op}_{\mathcal{A}}$ are its \emph{operations}. We may omit the subscript on operations when the algebra is understood. When we are thinking of an algebra as an algebra of processes, we may say `operator' rather than `operation.'
A homomorphism between two algebras is a map between their carriers respecting their operations; we  therefore have a category of $\Sigma$-algebras.

Given such a $\Sigma$-algebra, every term $t$ has a  \emph{denotation} $\dsem{t}(\rho)$, an element of the carrier, given an assignment $\rho$ of elements of the carrier to every variable; we  often confuse terms with their denotation. The algebra \emph{satisfies} an equation $t = u$ if $t$ and $u$ have the same denotation for every such assignment. If $\mathcal{A}$ satisfies all the axioms of a theory $\Th$, it is called a $\Th$-algebra; the  $\Th$-algebras form a subcategory of the category of $\Sigma$-algebras.  Any equation provable from the axioms of a theory $\Th$ is satisfied by any $\Th$-algebra. We say that a theory $\Th$ is \emph{(ground) equationally complete} with respect to a $\Th$-algebra if a (ground) equation is provable from $\Th$ if, and only if, it is  satisfied by the $\Th$-algebra.

Any finitary equational theory $\Th$ determines a free algebra monad $T_{\Th}$ on the category of sets, as well as  operations
	\[\mathrm{op}_{X}\type T_{\Th}(X)^{n} \rightarrow T_{\Th}(X)\]
	for any set $X$ and $\mathrm{op}: n \in \Sigma$, such that
	 $(T_{\Th}(X), (\mathrm{op}_{X}\type X^{n} \rightarrow X)_{\mathrm{op}: n \in \Sigma})$ is the  free $\Th$-algebra  over  $X$. Although	$T_{\Th}(X)$ is officially just a set, the carrier of the free algebra, we may also use $T_{\Th}(X)$ to denote the free algebra 
itself.
In the above example the monad is the finite powerset monad:
\[{\cal F}(X) = \{ u \subseteq X \mid u \mbox{
 is finite}\}\]
with $\boxempty_{X}$ and $\mystop_{X}$ being union and the empty set, respectively.

 \section{A first attempt at analysing CSP} \label{CSPconstructors1}

We consider the fragment of CSP with deadlock, action prefix, internal and external choice, relabelling and concealment, and concurrency and interleaving. Working over a fixed alphabet $A$ of \emph{actions}, we consider the following operation symbols: \vspace{1em}

{\bf Deadlock}
\[\mystop \type 0\]

{\bf Action}
\[a\myactionsub - \type 1 \quad\quad (a \in A)\]

{\bf Internal and External Choice}
\[\sqcap, \boxempty \type 2\]

{\bf Relabelling and Concealment}
\[f(-), - \backslash a \type 1\]
for any \emph{relabelling function} $f\type A \rightarrow A$ and action $a$. If $A$ is infinite, this makes the syntax infinitary; as that causes us no problems,  we do not avoid it.\vspace{1em}

{\bf Concurrency and Interleaving}
\[\mypar, \myinterleave \; \type 2\]

The signature of our (first) equational theory
$\mathrm{CSP}(\boxempty)$ for CSP only has operation symbols for the subset of
these  operators which are naturally thought of as constructors, namely
deadlock, action and internal and external choice. Its axioms are those given by de Nicola in~\cite{DeN85}. They are largely very natural and modular, and are as follows:
\begin{itemize}
		\item $\boxempty,\mystop$ is a semilattice with a zero
		(i.e., the above axioms for a semilattice with a zero).
		\item $\sqcap$ is a semilattice (i.e., the axioms
		stating the associativity, commutativity and
		 {idempotence} of $\sqcap$).
		\item  $\boxempty$  and $\sqcap$ distribute over each other:
		  \[x  \boxempty (y \sqcap z) =  (x \boxempty y) \sqcap (x  \boxempty z)
		  \quad\quad x \sqcap (y \boxempty z) =  (x \sqcap y) \boxempty (x \sqcap z)\]
		\item Actions distribute over $\sqcap$:
		\[a\myaction (x \sqcap y) = a\myaction x \sqcap a\myaction  y\]
		 and:
		\[a\myaction x \boxempty a\myaction y = a\myaction x \sqcap a\myaction y\]  
	\end{itemize}
All these axioms are mathematically natural except
the last which involves a  relationship between three different operators. 

We adopt some useful standard notational abbreviations. For $n \geq 1$ we write $\bigsqcap_{i= 1}^{n} t_{i}$ to abbreviate $t_{1} \sqcap \ldots \sqcap t_{n}$, intending $t_{1}$ when $n = 1$. We  assume that parentheses associate to the left; however as $\sqcap$ is associative, the choice does not matter. As $\sqcap$ is a semilattice, we can even index over nonempty finite sets, as in $\bigsqcap_{i \in I} t_{i}$, assuming some standard ordering of the $t_{i}$ without repetitions. As $\boxempty$ is a semilattice with a zero, we can adopt analogous notations $\bigbox_{i= 1}^{n} t_{i}$ and $\bigbox_{i \in I} t_{i}$ but now also allowing $n$ to be $0$ and $I$ to be $\emptyset$.

As $\sqcap$ is a semilattice we can define a partial order for which it is the greatest lower bound by writing $t \sqsubseteq u$ as an abbreviation for $t \sqcap u = t$; then, as $\boxempty$ distributes over $\sqcap$, it is monotone with respect to $\sqsubseteq$: that is, if $x \sqsubseteq x'$ and $y \sqsubseteq y'$ then $x \boxempty y \sqsubseteq x' \boxempty y'$. (We mean all this in a formal sense, for example, that if $t \sqsubseteq u$ and $u \sqsubseteq v$ are provable, so is $t \sqsubseteq v$, etc.)
We note the following, which is equivalent to the distributivity of $\sqcap$ over $\boxempty$, given  that $\sqcap$ and $\boxempty$ are semilattices, and the other distributivity, that  $\boxempty$ distributes over $\sqcap$:
\begin{equation} \label{helpful} x \sqcap (y \boxempty z)  =  x \sqcap (y \boxempty z)  \sqcap (x \boxempty y)
\end{equation}
The equation can also be written as $ x \sqcap (y \boxempty z)  \sqsubseteq (x \boxempty y)$.
Using this one can derive another helpful equation:
\begin{equation} \label{also helpful}
 (x \boxempty a\myaction z) \sqcap (y \boxempty a \myaction w)  =
 (x \boxempty a\myaction (z\sqcap w)) \sqcap (y \boxempty a \myaction (z\sqcap w))
\end{equation}

We next rehearse the original refusal sets model of CSP,
restricted to finite processes without divergence; this
provides a convenient context for identifying the initial model of
$\mathrm{CSP}(\boxempty)$ in terms of failures.

 A \emph{failure (pair)} is a pair $(w,W)$ with $w \in A^*$ and $W \subseteq_{\mathrm{fin}}A$. For every set $F$ of failure pairs, we define its set of \emph{traces} to be
\[\trF = \{w \mid \, (w,\emptyset) \in F\}\] 
and for every $w \in \trF$ we define its set of \emph{futures} to be: 
\[\futF(w) =  \{a \mid wa \in \trF\}\]
With that a \emph{refusal set} $F$ (aka a \emph{failure set}) is a set of failure pairs, satisfying the following conditions:
\begin{enumerate}
\item $\varepsilon \in \trF$
\item  $wa \in \trF \Rightarrow  w \in \trF$
\item $(w,W) \in F \wedge V \subseteq W \Rightarrow (w,V) \in F$
\item $(w,W) \in F \wedge a \notin \futF \Rightarrow (w, W \cup \{a\}) \in F$
\end{enumerate}
A refusal set is \emph{finitary} if its set of traces is finite.

The collection of finitary refusal sets can be turned into a
$\mathrm{CSP}(\boxempty)$-algebra $\fR$ by  the following standard
definitions of the operators:

\[\begin{array}{lcl}
	\mystop_{\fR}& = &\{(\varepsilon,W) \mid W \subseteq_{\mathrm{fin}} A\}\\
	a\myactionsub_{\fR} F & = &\{(\varepsilon, W) \mid a \notin W\} \cup
	                                     \{(aw,W)\mid (w,W) \in F\}\\
	 F \sqcap_{\fR} G & = & F \cup G\\
     F \boxempty_{\fR} G & = &
           \{(\varepsilon,W)\mid (\varepsilon,W) \in F \cap G\}
           \cup
           \{(w,W)\mid  w \neq \varepsilon,~ (w,W) \in F \cup G\}
\end{array}\]
The other CSP operation symbols also have standard interpretations
over the collection of {finitary} refusal sets: 
\[\begin{array}{lcl}
f(F)  & = &  \{(f(w), W) \mid (w,f^{-1}(W) \cap \fut_{F}(w)) \in F\}\\
F\backslash a  & =  &  \{(w \backslash a,W) \mid (w,W \cup \{a\}) \in F\}\\
F \mypar G & = & \{ (w, W \cup V) \mid (w,W) \in F,~ (w,V) \in G\}\\
F \myinterleave G & = & \{ (w, W) \mid (u,W) \in F,~ (v,W) \in G,~ w \in u\!\mid \!v\}
 \end{array}\]
with the evident action of $f$ on sequences and sets of actions, and
where $w\backslash a$ is obtained from $w$ by removing all occurrences
of $a$, and where $u\!\mid \!v$ is the set of interleavings of $u$ and $v$.

\begin{lemma} \label{Refusal_accounting} Let $F$ be a {\finitary} refusal set. Then for every  $w \in \trF$ there are $V_{1}, \ldots, V_{n} \subseteq \futF(w)$, including $\futF(w)$,  such that 
 $(w,W) \in F$ iff $W \cap V_{i} = \emptyset$ for some $i \in \{1, \ldots, n\}$.
\end{lemma}
\begin{proof} The closure conditions imply that
$(w,W)$ is in $F$  iff   $(w,W \cap \futF(w))$ is.
Thus we only need to be concerned about pairs $(w,W)$ with
$W \subseteq \futF(w)$.
Now, as $\futF(w)$ is finite, for any relevant $(w,W) \in F$, of which there are finitely many, we
can  take $V$ to be $\futF(w) \backslash W$, and we obtain finitely many such sets.
As $(w,\emptyset) \in F$, these include $\futF(w)$.
\qed \end{proof}
\begin{lemma} \label{Definability}
 {All {\finitary} refusal sets are definable by closed $\mathrm{CSP}(\boxempty)$ terms.} \end{lemma}
\begin{proof} Let $F$ be a {\finitary} refusal set. We proceed by induction on the length of the longest trace in $F$. 
By the previous lemma  there are sets $V_{1}, \ldots, V_{n}$, including $\futF(\varepsilon)$, such that
  $(\varepsilon,W) \in F$ iff $W \cap V_{i} = \emptyset$ for some $i \in \{1, \ldots, n\}$.
 Define $F_{a}$, for $a \in \futF(\varepsilon)$, by:
  \[F_{a} = \{ (w,W) \mid (aw,W) \in F\}\] 
  Then it is not hard to see that each $F_{a}$ is a {\finitary} refusal set, and that
  \[F = \bigsqcap_{i}\bigbox_{a \in V_{i}} a\myaction F_{a}\]
  As the longest trace in $F_{a}$ is strictly shorter than the longest one in $F$, the proof concludes, employing the induction hypothesis.
\qed \end{proof}

We next recall some material from de Nicola~\cite{DeN85}. Let $\mathcal{L}$ be a collection of sets; we say it is \emph{saturated} if whenever $L \subseteq L' \subseteq \bigcup \mathcal{L}$, for $L \in \mathcal{L}$ then $L' \in \mathcal{L}$. Then a closed $\mathrm{CSP}(\boxempty)$-term $t$ is in \emph{normal form} if it is of the form:
\[ \bigsqcap_{L \in \mathcal{L}}\bigbox_{a \in L} a\myaction t_{a} \]
where $\mathcal{L}$ is a finite non-empty saturated collection of finite sets of actions and each term $t_{a}$ is in normal form. Note that the concept of normal form is defined recursively.

\begin{proposition}\label{completeness}
  $\mathrm{CSP}(\boxempty)$ is ground equationally
  complete with respect to {$\fR$}.
\end{proposition}
\begin{proof}  Every term is provably equal in
  $\mathrm{CSP}(\boxempty)$ to a term in normal form. For the proof,
  follow that of Proposition A6 in~\cite{DeN85}; alternatively, it is
  a straightforward induction in which equation{s~(\ref{helpful})
  and~(\ref{also helpful}) are}
  helpful.  Further, it is an immediate consequence of Lemma 4.8
  in~\cite{DeN85} that if two normal forms have the same denotation in
  $\mathnormal{\fR}$ then they are identical (and
  Lemma~\ref{equal normal forms} below establishes a more general
  result). The result then follows. 
\qed \end{proof}

\begin{theorem} \label{Initial1} The {\finitary} refusal sets algebra
  $ \mathnormal{\fR}$ is the initial $\mathrm{CSP}(\boxempty)$ algebra.
\end{theorem}
\begin{proof}  Let the initial such algebra be $\mathrm{I}$. There is
  a unique homomorphism $h \type \mathrm{I} \rightarrow
  \mathnormal{\fR}$. By Lemma~\ref{Definability}, $h$ is a
  surjection.
By the previous proposition,    $\mathnormal{\fR}$ is complete for equations between closed terms, and so $h$ is  an injection. So $h$ is an isomorphism, completing the proof.
\qed \end{proof}

\section{Effect deconstructors} \label{Deconstructors}

In the algebraic theory of effects, the semantics of effect \emph{deconstructors}, such as exception handlers, is given using homomorphisms from free algebras. 
In this case we are interested in $T_{\mathrm{CSP}(\boxempty)}(\emptyset)$. This is the initial $\mathrm{CSP}(\boxempty)$ algebra, $\mathnormal{\fR}$, so given a $\mathrm{CSP}(\boxempty)$ algebra:
\[\mathcal{A} = (T_{\mathrm{CSP}(\boxempty)}(\emptyset), \sqcap_{\mathcal{A}},\mystop_{\mathcal{A}},
(a\myactionsub_{\mathcal{A}}), \boxempty_{\mathcal{A}})\] 
there is a unique homomorphism:
\[h \type \mathnormal{\fR} \rightarrow \mathcal{A}\]  

\runinhead{Relabelling}
We now seek to define $f(-)\type T_{\mathrm{CSP}(\boxempty)}(\emptyset) \rightarrow T_{\mathrm{CSP}(\boxempty)}(\emptyset)$ homomorphically. Define an algebra $\mathnormal{Rl}$ on 
$T_{\mathrm{CSP}(\boxempty)}(\emptyset)$ by putting, for refusal sets $F,G$:
\[\mystop_{\mathnormal{Rl}} = \mystop_{\mathnormal{\fR}} \]
\[(a\myactionsub_{\mathnormal{Rl}}  F)  = (f(a)\myactionsub_{\mathnormal{\fR}}  F)\]
\[F \sqcap_{\mathnormal{Rl}} G = F \sqcap_{\mathnormal{\fR}} G \quad\quad F \boxempty_{\mathnormal{Rl}} G = F \boxempty_{\mathnormal{\fR}} G\]
One has to verify this gives a $\mathrm{CSP}(\boxempty)$-algebra, which amounts to verifying that 
the two action equations hold, for example that, for all $F,G$:
\[a\myactionsub_{\mathnormal{Rl}}  (F \sqcap_{\mathnormal{Rl}} G) = 
(a\myactionsub_{\mathnormal{Rl}}  F)
\sqcap_{\mathnormal{Rl}} 
(a\myactionsub_{\mathnormal{Rl}}  G)
\]
which is equivalent to:
\[f(a)\myactionsub_{\mathnormal{\fR}}  (F \sqcap_{\mathnormal{\fR}} G) = 
(f(a)\myactionsub_{\mathnormal{\fR}}  F)
\sqcap_{\mathnormal{\fR}} 
(f(a)\myactionsub_{\mathnormal{\fR}}  G)
\]
We therefore have a unique homomorphism  \[\mathnormal{\fR} \stackrel{h_{\mathnormal{Rl}}}{\longrightarrow} \mathnormal{Rl}\] and so the following equations hold over the algebra $\mathnormal{\fR} $:
\[h_{\mathnormal{Rl}} (\mystop) = \mystop\]
\[h_{\mathnormal{Rl}} (a\myaction F) = f(a)\myaction h_{\mathnormal{Rl}}(F)\]
\[h_{\mathnormal{Rl}} (F \sqcap G) = h_{\mathnormal{Rl}}(F)\sqcap h_{\mathnormal{Rl}}(G) \quad\quad h_{\mathnormal{Rl}} (F  \boxempty G) = h_{\mathnormal{Rl}}(F) \boxempty h_{\mathnormal{Rl}}(G)\]
Informally one can use these equations to define $h_{\mathnormal{Rl}}$ by a `\emph{principle of equational recursion},'  but one must remember to verify that the implicit algebra obeys the required equations.

We use $h_{\mathnormal{Rl}}$ to interpret relabelling. We then immediately recover the familiar CSP laws:
\[f (\mystop) = \mystop\]
\[f (a\myaction x) = f(a)\myaction f(x)\]
\[f (x \sqcap y) = f(x)\sqcap f(y) \quad\quad f (x  \boxempty y) = f(x) \boxempty f(y)\]
which we now see to be restatements of the homomorphism of relabelling.

\runinhead{Concealment}
There is a difficulty here. We do not have that
\[  (F \boxempty G)\backslash a = F\backslash a \, \boxempty G\backslash a\]
but rather have the following two equations (taken from~\cite{DeN85}):
\begin{equation} \label{eqbox1}  ((a\myaction F) \boxempty G)\backslash  a = F\backslash  a \sqcap ((F \boxempty G)\backslash a) \end{equation}
\begin{equation} \label{eqbox2} (\bigbox_{i = 1}^{n} a_{i}\boxaction F_{i})\backslash a = 
\bigbox_{i = 1}^{n} a_{i} \boxaction(F_{i} \backslash a)\end{equation}
where no $a_{i}$ is $a$. Furthermore, there is no direct definition of concealment via an equational recursion, i.e., there is no suitable choice of algebra, $\boxempty_{\mathcal{A}}$ etc. For, if there were, we would have:
\begin{equation} \label{eq2} (F \boxempty G)\backslash a = F\backslash a \, \boxempty_{\mathcal{A}} G\backslash a
\end{equation}
So if $a$ does not occur in any trace of $F'$ or $G'$ we would have:
\[\begin{array}{lcl}
F' \, \boxempty_{\mathcal{A}} G'& = & F'\backslash a \, \boxempty_{\mathcal{A}} G'\backslash a\\
                                & = & (F' \boxempty G')\backslash a\\
                                & = & F' \boxempty G'
\end{array}\]
but, returning to equation~(\ref{eq2}), $a$ certainly does not occur in any trace of $F\backslash a$ or $G\backslash a$ and so we would have:
\[\begin{array}{lcl}
(F \boxempty G)\backslash a & = & F\backslash a \, \boxempty_{\mathcal{A}} G\backslash a\\
                  & = & F\backslash a \, \boxempty_{\mathnormal{\fR}} G\backslash a
\end{array}\]
which is false.
It is conceivable that although there is no direct homomorphic definition of concealment, there may be an indirect one where other functions (possibly with parameters---see below) are defined homomorphically and concealment is definable as a combination of those.

\subsection{Concurrency operators}

Before trying to recover from the difficulty with concealment, we look at a further difficulty, that of accommodating binary deconstructors, particularly parallel operators. We begin with a simple example in a strong bisimulation context, but rather than a  concurrency operator in the style of CCS we consider one analogous to CSP's $\mypar$.

We take as signature a unary action prefix, $a.-$, for $a \in A$, a nullary $\prog{NIL}$ and a binary sum $+$.  The axioms are that $+$ is a semilattice with zero $\prog{NIL}$; the initial algebra is then that of finite synchronisation trees $\mathnormal{ST}$.
Every synchronisation tree $\tau$ has a finite depth and can be written as  \[\sum_{i = 1}^{n} a_{i}.\tau_{i}\] for some $n \geq 0$,  where the $\tau_{i}$ are also synchronisation trees (of strictly smaller depth), and where no pair 
$(a_{i},\tau_{i})$ occurs twice. The order of writing the summands makes no difference to the tree denoted.

One can define a binary synchronisation operator $\mypar$ on synchronisation trees $\tau = \sum_{i} a_{i}.\tau_{i}$ and $\tau' = \sum_{j} b_{j}.\tau_{j} $ by induction on the depth of $\tau$ (or $\tau'$):
\[\tau\mypar\tau' = \sum_{a_{i} \,= \,b_{j}} a_{i}.(\tau_{i} \mypar \tau'_{j}) \]
Looking for an equational recursive definition of $\mypar$, one may try a `mutual (parametric) equational recursive definition' of $\mypar$ and a certain family $\mypar^{a}$  with $x,y,z$ varying over 
 $\mathnormal{ST}$:
\[\begin{array}{ccl}\prog{NIL} \mypar z &  =  & \prog{NIL}\\
(x + y) \mypar z &  =  & (x \mypar z) + (y \mypar z)\\
a.x \mypar z &  = & x \mypar^{a} z
\end{array}\]
and

\[\begin{array}{ccl} z \mypar^{a }\prog{NIL} &  =  & \prog{NIL}\\
z \mypar^{a } (x + y)  &  =  & (z \mypar^{a } x) + (z \mypar^{a } y)\\
z \mypar^{a } b.x  &  = &   \left\{
\begin{array}{ll}
a.(z \mypar x) & (\mbox{if  } b = a)\\
\prog{NIL} & (\mbox{if  } b \neq a)
\end{array}
\right.
\end{array}\]
Unfortunately, this definition attempt is not an equational recursion.
Mutual (parametric)  equational recursions are single ones to an algebra on a product. Here we wish a map: $ \mathnormal{ST} \rightarrow \mathnormal{ST} \times \mathnormal{ST}$. Informally we would write such clauses as:
\[\langle (x + y) \mypar z\; , \; z \mypar^{a} (x + y) \rangle = \langle (x\mypar z) + (y\mypar z)\;,\; (z\mypar^{a} x) + (z\mypar^{a} y) \rangle\]
with the recursion variables, here $x,y$,  on the left for $\mypar$ and on the right for $\mypar^{a}$. 
However:
\[\langle a.x \mypar z \; , \; z \mypar^{a} b.x \rangle =  \left\{
\begin{array}{ll}
\langle  x \mypar^{a} z\;,\;  a.(z \mypar x)  \rangle & (\mbox{if  } b = a)\\
\langle  x \mypar^{a} z\;,\; \prog{NIL} \rangle & (\mbox{if  } b \neq a)
\end{array}\right.
\]
does not respect this discipline: the recursion variable, here $x$,  (twice) switches places with the parameter $z$.
   
We are therefore caught in a dilemma. One can show, by induction on the depth of synchronisation trees, that the above definitions, viewed as equations for $\mypar$ and  
$\mypar^{a}$ have a unique solution: the expected synchronisation operator $\mypar$, and the functions $\mypar^{a}$ defined on synchronisation trees $\tau$ and $\tau' = \sum_{j} b_{j}.\tau_{j} $ by:
 \[\tau \mypar^{a} \tau' = \sum_{b_{j} \,= \, a}  a.(\tau\mypar \tau_{j})\]
 So we have a correct definition not in equational recursion format. So we must either
\begin{itemize}
\item find a different correct definition in the equational recursion format
\end{itemize}
or else
\begin{itemize}
\item find another algebraic format into which the correct definition fits.
\end{itemize}

When we come to the CSP parallel operator we do not even get as far as we did with synchronisation trees.
The problem is like that with concealment:  the distributive equation:
\[(F \boxempty F') \mypar G =  (F  \mypar G) \boxempty (F' \mypar G)\]
 does not hold. 
 One can show that there is no definition of $\mypar$ analogous to the above one for synchronisation trees, i.e., there is no suitable choice of algebra, $\boxempty_{\mathcal{A}}$ etc, and functions $\mypar^{a}$. The reason is that there is no binary operator $\boxempty'$ on ({\finitary}) failure sets such that, for all $F,G,H$  we have:
\[(F \boxempty F') \mypar G =  (F  \mypar G) \boxempty' (F' \mypar G)\]
For suppose, for the sake of contradiction, that there is such an operator. Then, fixing $F$ and $F'$,
choose $G$ such that $F\mypar G =  F$, $F'\mypar G =  F'$ and $(F \boxempty F') \mypar G =  (F\boxempty F')$. Then, substituting into the above equation, we obtain that $F \boxempty F' = F \boxempty' F'$ and so the above equation yields distributivity, which, in fact, does not hold.
As in the case of concealment, there may nonetheless be an indirect definition of $\mypar$.

A similar difficulty obtains for the CSP interleaving operator.
It too does not commute with $\boxempty$, and it too does not have any direct definition (the argument is like that for the concurrency operator but a little simpler, taking $G = \mystop$). As in the case of the concurrency operator,  there may be an indirect definition.

\section{Another choice of CSP effect constructors} \label{CSPconstructors2}

Equations (\ref{eqbox1}) and (\ref{eqbox2})
do not immediately suggest a recursive definition of concealment.
However, one can show that, for distinct actions $a_i$ ($i = 1,n$),
the following equation  holds between refusal sets:
\[(\bigbox_{i = 1}^{n} a_i\myaction F_i)\backslash a_j = 
             (F_j\backslash a_j) \sqcap ((F_j\backslash a_j) \boxempty
             \bigbox_{i \neq j} a_i\myaction (F_i\backslash a_j))   \]
where $1 \leq j \leq n$. Taken together with  equation (\ref{eqbox2}),  this  suggests a recursive definition in terms of deterministic external choice.
We therefore now change our choice of constructors, replacing binary external choice, action prefix and deadlock by  deterministic external choice.

So as our second signature for CSP we take a binary operation symbol $\sqcap$ of internal choice and, for any \emph{deterministic action sequence $\vec{a}$}
(i.e., any sequence of actions $a_i$ ($i = 1,n$), with the $a_i$ all different and $n \geq 0$),  an $n$-ary operation symbol $\bigbox_{\vec{a}}$ of deterministic external choice.
We write $\bigbox_{\vec{a}}(t_{1},\ldots,t_{n})$ as $\bigbox_{i = 1}^{n} a_i\boxaction t_i$ 
although it is more usual to use Hoare's notation $(a_{1} \rightarrow t_{1} \mid \dots \mid  a_{n} \rightarrow t_{n})$; we also use $\mystop$ to abbreviate $\bigbox_{\vec{a}}()$.
 
We have the usual semilattice axioms for $\sqcap$. 
Deterministic external choice is  commutative, in the sense that:
\[\bigbox_i a_i\boxaction x_i = \bigbox_i a_{\pi(i)}\boxaction x_{\pi(i)}\]
for any permutation $\pi$ of $\{1,\ldots,n\}$. Given this, we are justified in writing deterministic external choices over finite, possibly empty,  sets of actions, $\bigbox_{a \in I} a\boxaction t_{a}$, assuming some standard ordering of 
pairs $(a,t_{a})$ without repetitions.

For the next axiom it is convenient to write
$(a_1\myaction t_1) \boxempty \bigbox_{i= 2}^{n} a_i\boxaction t_i$
for $\bigbox_{i = 1}^{n} a_i\boxaction t_i$ (for $n \geq 0$).
The axiom states that deterministic external choice distributes
over internal choice:
\[(a_1\myaction  (x \sqcap x')) \;\boxempty \; \bigbox_{i= 2}^{n} a_i\boxaction x_i = \left((a_1\myaction  x) \;\boxempty \; \bigbox_{i= 2}^{n} a_i\boxaction x_i\right)\;\, \sqcap \;\,\left((a_1\myaction  x') \;\boxempty \; \bigbox_{i= 2}^{n} a_i\boxaction x_i\right)\]
This implies that deterministic external choice is monotone with respect to $\sqsubseteq$.

We can regard a, possibly nondeterministic,  external choice,
in which the $a_i$ need not be all different, as an abbreviation for a deterministic one, via:
\begin{eqnarray}\label{convention} \bigbox_i a_i\boxaction t_i  & = &
  \bigbox_{b \in \{a_1,\ldots, a_n\!\}} b\boxaction \left(\bigsqcap_{a_i = b}t_i\right)
  \end{eqnarray}
With that convention we may also write $a_1\myaction t_1 \;\boxempty \; \bigbox_{i= 2}^{n} a_i\boxaction t_i$ even when $a_1$ is some $a_i$, for $i > 1$.
We can now write our final axiom:
\begin{equation}\label{final axiom}
\left(\bigbox_i a_i\boxaction x_i\right) \;\sqcap\; \left((b_1\myaction y_1) \boxempty \bigbox_{j =2}^{n}  b_j\boxaction y_j\right)
\;\,\sqsubseteq\;\,
(b_1\myaction y_1) \;\boxempty\; \bigbox_i a_i\boxaction x_i
\end{equation}
Restricting the external choice $(b_1\myaction y_1) \boxempty
\bigbox_j  b_j\boxaction y_j$ to be deterministic gives an equivalent
axiom, as does restricting   $\bigbox_i a_i\boxaction x_i$ (in the presence of the others).

Let us call this equational theory $\mathrm{CSP}(\mid)$. The {\finitary} refusal sets form a  $\mathrm{CSP}(\mid)$-algebra $\mathnormal{\dfR}$ with the evident definitions:
\[\begin{array}{lcl}
F \sqcap_{\mathnormal{\dfR}} G & = & F \cup G\\
(\bigbox_{\vec{a}})_{\mathnormal{\dfR}}(F_{1},\ldots,F_{n}) & = & \{(\varepsilon, W) \mid W \cap \{a_{1},\dots,a_{n}\} = \emptyset\} \cup \{(a_{i}w,W) \mid (w,W) \in F_{i}\}
\end{array}\]

\begin{theorem}  The {\finitary} refusal sets algebra $\mathnormal{\dfR}$ is complete for equations between closed $\mathrm{CSP}(\mid)$ terms.
\end{theorem} 
\begin{proof} De Nicola's normal form can be regarded as written in the signature of \mbox{$\mathrm{CSP}(\mid)$}, and a straightforward induction proves that every $\mathrm{CSP}(\mid)$ term can be reduced to such a  normal form using the above axioms. But two such normal forms have the same denotation  whether they are regarded as $\mathrm{CSP}(\boxempty)$ or as $\mathrm{CSP}(\mid)$ terms, and in the former case, by Lemma 4.8 of~\cite{DeN85}, they are identical. 
\qed \end{proof}

\begin{theorem} \label{Initial2} The {\finitary} refusal sets algebra $\mathnormal{\dfR}$ is the initial $\mathrm{CSP}(\mid)$ algebra. 
\end{theorem}
\begin{proof} Following the proof of Lemma~\ref{Definability} we see that every {\finitary} refusal set is definable by a closed $\mathrm{CSP}(\mid)$ term. With that, initiality follows from the above completeness theorem, as in the proof of Theorem~\ref{Initial1}.
\qed \end{proof}

Turning to the deconstructors, relabelling again has a straightforward homomorphic definition:  given a relabelling function  $f\type A \rightarrow A$, 
$h_{\mathnormal{Rl}}\type  T_{\mathrm{CSP}(\mid)}(\emptyset) \rightarrow T_{\mathrm{CSP}(\mid)}(\emptyset)$ is defined homomorphically by:
\[h_{\mathnormal{Rl}}(F \sqcap G) = h_{\mathnormal{Rl}}(F) \sqcap h_{\mathnormal{Rl}}(G)\]
\[h_{\mathnormal{Rl}}(\bigbox_{i} a_{i}\boxaction F_{i}) = \bigbox_{i} f(a_{i})\boxaction h_{\mathnormal{Rl}}(F_{i})\]
As always one has to check that the implied algebra satisfies the equations, here those of $\mathrm{CSP}(\mid)$.

There is also now a natural homomorphic definition of concealment, $- \backslash a$, but, surprisingly perhaps, one needs to assume that $\boxempty$ is available. For every $a \in A$ one defines $h_{a} \type  T_{\mathrm{CSP}(\mid)}(\emptyset) \rightarrow T_{\mathrm{CSP}(\mid)}(\emptyset)$ homomorphically by:
\[h_{a}(F \sqcap G) = h_{a}(F) \sqcap h_{a}(G)\]
\[h_{a}\left(\bigbox_{i = 1}^{n} a_{i}\boxaction F_{i}\right) =  \left\{
\begin{array}{ll}
h_{a}(F_{j})  \sqcap (h_{a}(F_{j}) \boxempty \bigbox_{i \neq j} a_{i}\boxaction h_{a}(F_{i})) & (\mbox{if  } a = a_{j},\, \mbox{where} \,1 \leq j \leq n)\\
\bigbox_{i = 1}^{n} a_{i}\boxaction h_{a}(F_{i}) & (\mbox{if  } a \neq \mbox{any} \; a_{i})
\end{array}
\right.
\]
Verifying that the implicit algebra obeys satisfies the required equations is quite a bit of work. We record the result, but omit the calculations:
\begin{proposition} \label{Conceal} One can define a $\mathrm{CSP}(\mid)$-algebra $\mathnormal{Con}$ on $T_{\mathrm{CSP}(\mid)}(\emptyset)$ by:  
\[F \sqcap_{\mathnormal{Con}} G  = F \sqcap G\]
\begin{center}
$(\bigbox_{\vec{a}})_{\mathnormal{Con}}(F_{1},\ldots,F_{n})  =  \left\{
\begin{array}{ll}
F_{j}  \sqcap (F_{j} \boxempty \bigbox_{i \neq j} a_{i}\boxaction F_{i}) & (\mbox{if  } a = a_{j})\\
\bigbox_{i} a_{i}\boxaction F_{i} & (\mbox{if  } a \neq \mbox{any} \; a_{i})
\end{array}
\right.
$
\end{center}
\end{proposition}

The operator $\boxempty$ is, of course, no longer available as a constructor. However, it can alternatively be treated as a binary deconstructor. While its treatment as such is  no more successful than our treatment of parallel operators, it is also  no less successful.  We define it simultaneously with $(n\mathord+1)$-ary functions 
$\boxempty^{a_{1}\ldots a_{n}}$ on $T_{\mathrm{CSP}(\mid)}(\emptyset)$, for $n\geq 0$, where the $a_{i}$ are all distinct. That we are defining infinitely many functions simultaneously arises from dealing with the infinitely many deterministic choice operators (there would be be infinitely many even if we considered them as parameterised on the $a$'s). However, we anticipate that this will cause no real difficulty, given that we have overcome the difficulty of dealing with binary deconstructors.   

Here are the required definitions:
\begin{eqnarray} \label{box-def} (F \sqcap F') \boxempty G & = &   (F \boxempty G) \sqcap (F' \boxempty G)\nonumber\\
(\bigbox_{i} a_{i}\boxaction F_{i}) \boxempty G  & = &  (F_{1},\ldots, F_{n})\boxempty^{a_{1}\ldots a_{n}} G\nonumber\\
&&\nonumber\\
(F_{1},\ldots, F_{n})\boxempty^{a_{1}\ldots a_{n}} (G \sqcap G')  & = &   ((F_{1},\ldots, F_{n})\boxempty^{a_{1}\ldots a_{n}} G) \sqcap ((F_{1},\ldots, F_{n})\boxempty^{a_{1}\ldots a_{n}} G')\nonumber\\
(F_{1},\ldots, F_{n})\boxempty^{a_{1}\ldots a_{n}} (\bigbox_{j} b_{j}\boxaction G_{j}) & = & 
   (a_{1}\myaction F_{1}) \boxempty (\ldots  ((a_{n}\myaction F_{n}) \boxempty \bigbox_{j} b_{j}\boxaction G_{j})\dots)
   \end{eqnarray}
where, in the last equation, the notational  convention
$(a_1\myaction t_1) \boxempty  \bigbox_{i= 2}^{n} a_i\boxaction t_i$ {is used} $n$ times.
It is clear that $\boxempty$ together with the functions 
\[\boxempty^{a_{1}\ldots a_{n}}\type T_{\mathrm{CSP}(\mid)}(\emptyset)^{n+1}\rightarrow T_{\mathrm{CSP}(\mid)}(\emptyset)\]
 defined by:
\begin{equation} \label{box-super-def} \boxempty^{a_{1}\ldots
      a_{n}}(F_{1},\ldots,F_{n},G)  = (\bigbox_{i} a_{i}\boxaction F_{i}) \boxempty G
\end{equation}
satisfy the equations, and, using the fact that all {\finitary} refusal sets are definable by normal forms, one sees that they are the unique such functions.

We can treat the CSP parallel operator $\mypar$ in a similar vein following the pattern given above for parallel merge operators in the case of synchronisation trees.
We define it simultaneously with $(n\mathord+1)$-ary functions 
$\mypar^{a_{1}\ldots a_{n}}$ on $T_{\mathrm{CSP}(\mid)}(\emptyset)$, for $n\geq 0$, where the $a_{i}$ are all distinct:
\begin{eqnarray}  \label{par-def}(F \sqcap F') \mypar G & = &  (F \mypar G) \sqcap (F' \mypar G)\nonumber\\
 (\bigbox_{i} a_{i}\boxaction F_{i}) \mypar G & = & (F_{1},\ldots, F_{n})\mypar^{a_{1}\ldots a_{n}} G\nonumber\\
 \nonumber\\
 (F_{1},\ldots, F_{n})\mypar ^{a_{1}\ldots a_{n}} (G \sqcap G') & = &  ((F_{1},\ldots, F_{n})\mypar^{a_{1}\ldots a_{n}} G) \sqcap ((F_{1},\ldots, F_{n})\boxempty^{a_{1}\ldots a_{n}} G')\nonumber\\
 (F_{1},\ldots, F_{n})\mypar^{a_{1}\ldots a_{n}} (\bigbox_{j} b_{j}\boxaction G_{j}) & = & 
   \bigbox_{a_{i}  =  b_{j}} a_{i}\boxaction(F_{i}\mypar G_{j})
\end{eqnarray}
Much as before,  $\mypar$ together with the functions $\mypar^{a_{1}\ldots a_{n}}\type T_{\mathrm{CSP}(\mid)}(\emptyset)^{n+1}\rightarrow T_{\mathrm{CSP}(\mid)}(\emptyset)$ defined by:
\[\mypar^{a_{1}\ldots a_{n}}(F_{1},\ldots,F_{n},G)=
(\bigbox_{i} a_{i}\boxaction F_{i}) \mypar G\]
are the unique functions satisfying the equations.

Finally we consider the CSP interleaving operator $\myinterleave$. We define this by following an idea, exemplified in the ACP literature~\cite{BK85,BK86}, of splitting an associative operation into several parts. Here we split $\myinterleave$ into a \emph{left interleaving} operator $\myinterleave^{l}$ and a  \emph{right interleaving} operator $\myinterleave^{r}$ so that: 
\[F \myinterleave G =  (F \myinterleave^{l} G) \boxempty (F \myinterleave^{r} G)\]
In ACP the parallel operator is split into three parts: a left merge, a right merge (defined in terms of the left merge), and a communication merge; in a subtheory, PA, there is no communication, and the parallel operator, now an interleaving one, is split into left and right parts~\cite{BK86}. The idea of splitting an associative operation into several operations can be found in a much wider context~\cite{EFG08} where the split into two  or three parts is axiomatised by the respective notions of dendriform dialgebra and trialgebra.

 Our left and right interleaving are defined by the following  `binary deconstructor' equations:
\begin{eqnarray}  \label{inter-def}
(F \sqcap F') \myinterleave^{l}  G & = &  (F \myinterleave^{l}  G) \sqcap (F' \myinterleave^{l}  G)\nonumber\\
(\bigbox_{i = 1}^{n} a_{i}\boxaction F_{i}) \myinterleave^{l}  G & = &  \bigbox_{i} a_{i}\boxaction ((F_{i} \myinterleave^{l}  G) \boxempty (F_{i} \myinterleave^{r}  G))\nonumber\\
\nonumber\\
 G \myinterleave^{r} (F \sqcap F')  & = &  (G \myinterleave^{r} F) \sqcap (G \myinterleave^{r} F')\nonumber\\
G \myinterleave^{r} (\bigbox_{i = 1}^{n} a_{i}\boxaction F_{i}) & = & 
\bigbox_{i} a_{i}\boxaction ((G \myinterleave^{l}  F_{i}) \boxempty (G \myinterleave^{r}  F_{i})) 
\end{eqnarray}

As may be expected, these equations also have unique solutions, now given by:
\[\begin{array}{lcl}
F \myinterleave^{l} G & = & \{(\varepsilon,W) \mid (\varepsilon,W) \in
F\} \cup \{ (w, W) \mid (u,W) \in F,~  (v,W) \in G,~ w \in u\!\mid^{l}
\!v\}\\
F \myinterleave^{r} G & = &  \{(\varepsilon,W) \mid (\varepsilon,W)
\in G\} \cup  \{ (w, W) \mid (u,W) \in F,~ (v,W) \in G,~  w \in
u\!\mid^{r} \!v\} 
 \end{array}\]
 where $u\!\mid^{l} \!v$  is the set of interleavings of $u$ and $v$ which
 begin with  a letter of $u$, and $u\!\mid^{r} \!v$ is defined analogously.
It is interesting to note that:
\[F \myinterleave^{l} (G \sqcap G') =  (F \myinterleave^{l} G) \sqcap (F \myinterleave^{l}  G')\]
and similarly for $\myinterleave^{r}$.

\section{Adding divergence}\label{divergence}

The treatment of CSP presented thus far dealt with finite
divergence-free processes only. There are several ways to extend the
refusal sets model of Section~\ref{CSPconstructors1} to infinite
processes with divergence. The most well-known model is the
\emph{failures/divergences} model of \cite{Hoa85}, further elaborated in
\cite{Ros98}. A characteristic property of this model is that
divergence, i.e.,  an infinite sequence of internal actions, is
modelled as {\it Chaos}, a process that satisfies the equation:
\begin{eqnarray} \label{chaos}
\textit{Chaos} \boxempty x & = &  \textit{Chaos} \sqcap x = \textit{Chaos}
\end{eqnarray}
So after {\it Chaos} no further process activity is discernible.

An alternative extension is the \emph{stable failures} model proposed
in \cite{BKO87}, and also elaborated in \cite{Ros98}. This model
equates processes that allow the same \emph{observations}, where actions
and deadlock are considered observable, but divergence does not give
rise to any observations. A failure pair $(w,W)$---now allowing
  $W$ to be infinite---records an
observation in which $w$ represents a sequence of actions being
observed, and $W$ represents the observation of deadlock under the
assumption that the environment in which the observed process is
running allows {only} the (inter)actions in the set $W$. Such an
observation can be made if after engaging in the sequence of visible
actions $w$, the observed process reaches a state in which no further
internal actions are possible, nor any actions from the set $W$.
Besides failure pairs, also traces are observable, and thus the
observable behaviour of a process is given by a pair $(T,F)$ where
$T$ is a set of traces and $F$ is a set of failure pairs.
Unlike the model $\fR$ of Section~\ref{CSPconstructors1}, the traces
are not determined by the failure pairs. In fact, in a process that
can diverge in every state, the set of failure pairs is empty, yet the
set of traces conveys important information.

In the remainder of this paper we add a constant $\Omega$ to the
signature of CSP that is a zero for the semilattice generated by $\sqcap$.
This will greatly facilitate the forthcoming development.
Intuitively, one may think of $\Omega$ as divergence in the stable
failures model.

W.r.t.\ the equational theory  $\mathrm{CSP}(\boxempty)$ of
Section~\ref{CSPconstructors1} we thus add the constant $\Omega$ and
the single axiom:
\vspace{-1em}
\begin{eqnarray} \label{Omega} x \sqcap \Omega = x
\end{eqnarray}
thereby obtaining the theory $\mathrm{CSP}(\boxempty,\Omega)$.
We note two useful derived equations:
\begin{eqnarray}  \label{Omega-derived}
x \sqcap (\Omega \boxempty y) & = &  x \sqcap (x \boxempty y)\nonumber\\
(\Omega \boxempty x) \sqcap (\Omega \boxempty y) & = & (\Omega \boxempty x) \boxempty (\Omega \boxempty y) 
\end{eqnarray}

Semantically, a \emph{process} is now given by a pair $(T,F)$, where $T$ is a set of traces and $F$ is a set of failure pairs that satisfy the following conditions:
\begin{enumerate}\parskip 0pt
\item $\varepsilon \in T$
\item $wa \in T \Rightarrow  w \in T$
\item $(w,W) \in F \Rightarrow w \in T$
\item $(w,W) \in F \wedge V \subseteq W \Rightarrow (w,V) \in F$
\item $(w,W) \in F \wedge \forall a \in V.\, wa \notin T \Rightarrow (w, W \cup V) \in F
\qquad\qquad (\mbox{where $V \subseteq A$})$
\end{enumerate}
The two components of such a pair $P$ are denoted $T_P$ and $F_P$,
respectively, and for $w \in T_P$ we define $\futP(w) := \{a \in A \mid wa\in T_P\}$.
We can define the CSP operators on processes by setting 
$$P \mathrel{\mathrm{op}} Q = (P \mathrel{\mathrm{op}_\T} Q, P \mathrel{\mathrm{op}_\R} Q)$$
where $\mathrm{op}_\T$ is given by:
\[\begin{array}{lcl}
	\mystop_{\mathnormal{\T}}& = &\{\varepsilon\}\\
	a\myactionsub_{\mathnormal{\T}} P & = & \{\varepsilon\} \cup \{aw \mid w \in T_P\}\\
	 P \sqcap_{\mathnormal{\T}} Q & = & T_P \cup T_Q\\
     P \boxempty_{\mathnormal{\T}} Q & = & T_P \cup T_Q\\
f_{\T}(P)  & = &  \{f(w) \mid w \in T_P\}\\
P\backslash_{\T} a  & =  &  \{w \backslash a \mid w \in T_P\}\\
P \mypar_{\T} Q & = & \{ w \mid w \in T_P,~ w \in T_Q\}\\
P \myinterleave_{\T} Q & = & \{ w \mid u \in T_P,~ v \in T_Q,~ w \in u\!\mid \!v\}
 \end{array}\]
and {$\mathrm{op}_\R$ is given as $\mathrm{op}_\fR$} was in Section~\ref{CSPconstructors1},
but without the restriction to finite sets $W$ in defining $\mystop_\R$.
For the new process $\Omega$ we set
\[\Omega_\T = \{\varepsilon\} 
\qquad \mbox{and} \qquad
\Omega_\R = \emptyset\]
This also makes the collection of processes into a
$\mathrm{CSP}(\boxempty,\Omega)$-algebra, $\F$.

A process $P$ is called \emph{finitary} if $T_P$ is finite.
The {\finitary} processes evidently form a subalgebra of
$\mathnormal{\F}$; we call it $\mathnormal{\fF}$.

\begin{lemma} \label{Refusal_accounting_Omega} Let $P$ be a
  {\finitary} process. Then, for every  $w \in T_P$ there is an $n \geq 0$ and
  $V_{1}, \ldots, V_{n} \subseteq \futF(w)$
  such that
 $(w,W) \in F_P$ iff $W \cap V_{i} = \emptyset$ for some $i \in \{1, \ldots, n\}$.
\end{lemma}
\begin{proof} Closure conditions 4 and 5 above imply that
$(w,W) \in F_P$  if, and only if,    $(w,W \cap \futP(w)) \in F_P$.
Thus we only need to be concerned about pairs $(w,W)$ with
$W \subseteq \futP(w)$.
Now, as $\futP(w)$ is finite, for any relevant $(w,W) \in F$, of which there are finitely many, we
can  take $V$ to be $\futP(w) \backslash W$, and we obtain finitely many such sets.
\qed \end{proof}
Note that it may happen that $n=0$, in contrast with the case of Lemma~\ref{Refusal_accounting}.
\begin{lemma} \label{Definability_Omega}
All {\finitary} processes are definable by closed $\mathrm{CSP}(\boxempty,\Omega)$ terms.
\end{lemma}
\begin{proof} Let $P$ be a {\finitary} process. We proceed by
  induction on the length of the longest trace in $T_P$.
By the previous lemma  there are sets $V_{1}, \ldots, V_{n}$, for some
$n\geq 0$, such that
  $(\varepsilon,W) \in F$ iff $W \cap V_{i} = \emptyset$ for some $i \in \{1, \ldots, n\}$.
 Define $T_a$ and $F_{a}$, for $a \in T_P$, by:
  \[T_{a} = \{ w \mid aw \in T_P\}
\qquad
    F_{a} = \{ (w,W) \mid (aw,W) \in F_P\}\] 
  Then it is not hard to see that each $P_a := (T_a,F_{a})$ is a {\finitary} process, and that
  \[P = \left(\bigsqcap_{i}\bigbox_{a \in V_{i}} a\myaction P_{a}\right)
  ~\sqcap~ \left(\Omega \boxempty \bigbox_{a \in T_P} a\myaction P_{a}\right)\]
  As the longest trace in $T_{a}$ is strictly shorter than the longest
  one in $T_P$, the proof concludes, employing the induction hypothesis.
\qed \end{proof}
   
\begin{proposition} $\mathrm{CSP}(\boxempty,\Omega)$ is ground equationally
  complete with respect to both $\mathnormal{\F}$  and $\mathnormal{\fF}$.
\end{proposition}
\begin{proof} This time we recursively define a normal form as a
  $\mathrm{CSP}(\boxempty,\Omega)$-term of the form
\[ \bigsqcap_{L \in \mathcal{L}}\bigbox_{a \in L} a\myaction t_{a} 
\qquad \mbox{or} \qquad \Omega \boxempty \bigbox_{a \in K} a\myaction t_{a} \]
  where $\mathcal{L}$ is a finite non-empty saturated collection of finite sets of
  actions, $K$ is a finite set of actions, and each term $t_{a}$ is in normal form. 
  Every term is provably equal in $\mathrm{CSP}(\boxempty,\Omega)$ to
  a term in normal form;
  the proof proceeds as for Proposition~\ref{completeness},
  but now also using the derived equations~{(\ref{Omega-derived})}.
  Next, by Lemma~\ref{equal normal forms} below, if
  two normal forms have the same denotation in $\mathnormal{\F}$ then
  they are identical.
 So the result follows
 for $\mathnormal{\F}$, and then for $\mathnormal{\fF}$ too,
  as all closed terms denote {\finitary} processes.
\qed \end{proof}

\begin{theorem} \label{Initial1_Omega} The algebra $\mathnormal{\fF}$
  of finitary processes is the initial $\mathrm{CSP}(\boxempty,\Omega)$ algebra.
\end{theorem}
\begin{proof}  Let the initial such algebra be $\mathrm{I}$. There is
  a unique homomorphism $h \type \mathrm{I} \rightarrow
  \mathnormal{\fF}$. By Lemma~\ref{Definability_Omega}, $h$ is a surjection.
By the previous proposition, $\mathnormal{\fF}$ is complete for
equations between closed terms, and so $h$ is  an injection. Hence $h$ is
an isomorphism, completing the proof.
\qed \end{proof}

As in Section~\ref{CSPconstructors2}, in order to deal with deconstructors, particularly hiding, we replace external choice by deterministic external choice. The availability of $\Omega$ permits useful  additional such operators. The equational theory $\mathrm{CSP}(|,\Omega)$ has as signature the
binary operation symbol $\sqcap$, and for any deterministic action sequence
$\vec{a}$, the $n$-ary operation symbols $\bigbox_{\vec{a}}$ (as in
Section~\ref{CSPconstructors2}), as well as the new $n$-ary operation symbols
\raisebox{0pt}[0pt][0pt]{$\bigbox^\Omega_{\vec{a}}$},
for $n\geq 0$, which denote a
deterministic external choice with $\Omega$ as one of the summands.
We adopt conventions for \raisebox{0pt}[0pt][0pt]{$\bigbox^\Omega_{\vec{a}}$} analogous to those previously introduced for \raisebox{0pt}[0pt][0pt]{$\bigbox_{\vec{a}}(t_{1},\ldots,t_{n})$}. We write \raisebox{0pt}[0pt][0pt]{$\bigbox^\Omega_{\vec{a}}(t_{1},\ldots,t_{n})$}
as $\Omega\boxempty\bigbox_{i = 1}^{n} a_i\boxaction t_i$.
We also write $\Omega \boxempty
  (c_1\myaction t_1) \boxempty \bigbox_{j =2}^{n}c_{j}t_{j} $ for $\Omega \boxempty  \bigbox_{j =1}^{n}c_{j}t_{j} $, so that the $c_{j}$ ($j = 1,n$) must all be distinct. 

The first three groups of axioms of $\mathrm{CSP}(|,\Omega)$ are:
\begin{itemize}\parskip 0pt
\item $\sqcap,\Omega$ is a semilattice with a zero---here $\Omega$ is
  the 0-ary case of $\bigbox^\Omega_{\vec{a}}$,
\item both deterministic external choice operators $\bigbox_{\vec{a}}$
and \raisebox{0pt}[0pt][0pt]{$\bigbox^\Omega_{\vec{a}}$} are
commutative, as explained in Section~\ref{CSPconstructors2}, and
\item both deterministic external choice operators distribute over
  internal choice, as explained in Section~\ref{CSPconstructors2},
\end{itemize}
 
 Given commutativity, we are, as before, justified in writing
 deterministic external choices $\bigbox_{a \in I} a\boxaction t_{a}$
 or $\Omega \boxempty \bigbox_{a \in I} a\boxaction t_{a}$, over
 finite, possibly empty, sets of actions $I$, assuming some standard ordering of
pairs $(a,t_{a})$ without repetitions. Next, using the analogous convention to~(\ref{convention}) we can then also understand $\Omega \boxempty  \bigbox_{j =1}^{n}c_{j}t_{j} $, and so also $\Omega \boxempty
  (c_1\myaction t_1) \boxempty \bigbox_{j =2}^{n}c_{j}t_{j} $, even
  when the $c_{j}$ are not all distinct.
  With these conventions established, we can now state  the final group of axioms. These are all variants of
  Axiom (\ref{final axiom}) of Section~\ref{CSPconstructors2}, allowing
  each of the two deterministic external choices to have an $\Omega$-summand:
\begin{equation}
\left(\Omega \boxempty \bigbox_i a_i\boxaction x_i \right) \;\sqcap\; \left(\Omega \boxempty
  (b_1\myaction y_1) \boxempty \bigbox_{j =2}^{n}  b_j\boxaction y_j \right)
\;\,\sqsubseteq\;\,
 \Omega \boxempty (b_1\myaction y_1) \boxempty \bigbox_i a_i\boxaction x_i
\nonumber
\end{equation}
\begin{equation}
\left(\Omega \boxempty \bigbox_i a_i\boxaction x_i \right) \;\sqcap\; \left(
  (b_1\myaction y_1) \boxempty \bigbox_{j =2}^{n}  b_j\boxaction y_j \right)
\;\,\sqsubseteq\;\,
 \Omega \boxempty (b_1\myaction y_1) \boxempty \bigbox_i a_i\boxaction x_i
\nonumber
\end{equation}
\begin{equation}
\left(\bigbox_i a_i\boxaction x_i \right) \;\sqcap\; \left(\Omega \boxempty
  (b_1\myaction y_1) \boxempty \bigbox_{j =2}^{n}  b_j\boxaction y_j \right)
\;\,\sqsubseteq\;\,
  (b_1\myaction y_1) \boxempty \bigbox_i a_i\boxaction x_i
\nonumber
\end{equation}
\begin{equation}\label{final group}
\left(\bigbox_i a_i\boxaction x_i \right) \;\sqcap\; \left(
  (b_1\myaction y_1) \boxempty \bigbox_{j =2}^{n}  b_j\boxaction y_j \right)
\;\,\sqsubseteq\;\,
  (b_1\myaction y_1) \boxempty \bigbox_i a_i\boxaction x_i
\end{equation}
As in the case of Axiom~(\ref{final axiom}), restricting any of these choices to be deterministic results in an axiom of equivalent power.
 We note two useful derived equations:
 
 \begin{eqnarray} \label{derived-2}
  \bigbox_i a_i\boxaction x_i
 \sqcap
 (\Omega \boxempty \bigbox_j b_j\boxaction y_j)
&  = &
  \bigbox_i a_i\boxaction x_i
 \sqcap
 (\bigbox_i a_i\boxaction x_i \boxempty \bigbox_j b_j\boxaction y_j)\nonumber\\
 (\Omega \boxempty \bigbox_i a_i\boxaction x_i)
 \sqcap
(\Omega \boxempty \bigbox_j b_j\boxaction y_j)
 &  = &
 (\Omega \boxempty \bigbox_i a_i\boxaction x_i)  \boxempty \bigbox_j b_j\boxaction y_j
 \end{eqnarray}
where two further notational conventions are employed: 
 $ (\bigbox_{i = 1}^{m} a_i\boxaction t_i) \boxempty 
    (\bigbox_{j= 1}^{n} b_j\boxaction t'_j)$
stands for
 $ \bigbox_{k = 1}^{m+ n} c_k\boxaction t''_{k}$ 
  where $c_{k} = a_{k}$ and $t''_{k} = t_{k}$, for $k = 1,m$,
  and  $c_{k} = b_{k-m}$, and $t''_{k} = t'_{k-m}$, for $k = m\mathord+ 1,m \mathord+n$; 
 and  $ (\Omega \boxempty \bigbox_{i = 1}^{m} a_i\boxaction t_i) \boxempty 
    (\bigbox_{j= 1}^{n} b_j\boxaction t'_j)$ is understood analogously.
In fact, the first three axioms of (\ref{final group}) are also
derivable from (\ref{derived-2}), in the presence of the other
axioms, and thus may be replaced by (\ref{derived-2}).

The collection of processes is turned into a $\mathrm{CSP}(|,\Omega)$-algebra
$\dF$ as before, writing:
\[P \mathrel{\mathrm{op}_{\dF}} Q = (P \mathrel{\mathrm{op}_{\dT}} Q, P
 \mathrel{\mathrm{op}_{\dR}} Q)\]
and defining $\mathrm{op}_{\dT}$ and $\mathrm{op}_\dR$
in the evident way:
\[\begin{array}{lcl}
P \sqcap_{\mathnormal{\dT}} Q & = & T_P \cup T_Q\\
(\bigbox_{\vec{a}})_{\mathnormal{\dT}}(P_{1},\ldots,P_{n}) & = &
\{\varepsilon\} \cup \{a_{i}w \mid w \in T_{P_i}\}\\
(\bigbox^\Omega_{\vec{a}})_{\mathnormal{\dT}}(P_{1},\ldots,P_{n}) & = &
\{\varepsilon\} \cup \{a_{i}w \mid w \in T_{P_i}\}\\
(\bigbox^\Omega_{\vec{a}})_{\mathnormal{\dR}}(P_{1},\ldots,P_{n}) & = &
\{(a_{i}w,W) \mid (w,W) \in F_{P_i}\}
\end{array}\]
with $\sqcap_\dR$ and $(\bigbox_{\vec{a}})_{\mathnormal{\dR}}$
given just as in Section~\ref{CSPconstructors2}. 
Exactly as in Section~\ref{CSPconstructors2}, but now using the derived equations (\ref{derived-2}), we obtain:
\begin{theorem}  The algebra $\mathnormal{\dF}$ is complete for
  equations between closed $\mathrm{CSP}(|,\Omega)$ terms.
\end{theorem} 
\begin{theorem} \label{Initial2_Omega} The {\finitary} subalgebra $\dfF$ of
  $\mathnormal{\dF}$ is the initial $\mathrm{CSP}(|,\Omega)$ algebra.
\end{theorem}

Turning to the deconstructors, relabelling and concealment can again be treated homomorphically.  For relabelling by $f$ one simply adds the equation:
\[h_{\mathnormal{Rl}}(\Omega \boxempty\bigbox_{i} a_{i}\boxaction F_{i}) = \Omega \boxempty \bigbox_{i} f(a_{i})\boxaction h_{\mathnormal{Rl}}(F_{i})\]
to the treatment in Section~\ref{CSPconstructors2}, and checks that the implied algebra satisfies the equations.
Pleasingly, the treatment of concealment can be simplified in such a way that
the deconstructor $\boxempty$ is no longer needed.  For every $a \in A$ one defines $h_{a} \type  T_{\mathrm{CSP}(\mid,\Omega)}(\emptyset) \rightarrow T_{\mathrm{CSP}(\mid,\Omega)}(\emptyset)$ homomorphically by:
\[h_{a}(P \sqcap Q) = h_{a}(P) \sqcap h_{a}(Q)\]
\[h_{a}\left(\bigbox_{i = 1}^{n} a_{i}\boxaction P_{i}\right) =  \left\{
\begin{array}{ll}
h_{a}(P_{j})  \sqcap (\Omega \boxempty \bigbox_{i \neq j} a_{i}\boxaction h_{a}(P_{i})) & (\mbox{if  } a = a_{j},\, \mbox{where} \,1 \leq j \leq n)\\
\bigbox_{i = 1}^{n} a_{i}\boxaction h_{a}(P_{i}) & (\mbox{if  } a \neq \mbox{any} \; a_{i})
\end{array}
\right.
\]
\[h_{a}\left(\Omega \boxempty \bigbox_{i = 1}^{n} a_{i}\boxaction P_{i}\right) =  \left\{
\begin{array}{ll}
h_{a}(P_{j})  \sqcap (\Omega \boxempty \bigbox_{i \neq j} a_{i}\boxaction h_{a}(P_{i})) & (\mbox{if  } a = a_{j},\, \mbox{where} \,1 \leq j \leq n)\\
\Omega \boxempty \bigbox_{i = 1}^{n} a_{i}\boxaction h_{a}(P_{i}) & (\mbox{if  } a \neq \mbox{any} \; a_{i})
\end{array}
\right.
\]
Note the use of the new form of deterministic choice here. 
One has again to verify that the implicit algebra obeys satisfies the required equations. 
The treatment of the binary deconstructors
$\boxempty$, $\mypar$ and $\myinterleave$ is also a trivial adaptation of
the treatment in Section~\ref{CSPconstructors2}. For $\boxempty$ one adds a further 
auxiliary operator $\boxempty^{\Omega,a_{1}\ldots a_{n}}$ and the equations:
\begin{eqnarray}
(\Omega \boxempty \bigbox_{i} a_{i}\boxaction P_{i}) \boxempty Q  & = &  (P_{1},\ldots, P_{n})\boxempty^{\Omega,a_{1}\ldots a_{n}} Q\nonumber\\
 (P_{1},\ldots, P_{n})\boxempty^{\Omega,a_{1}\ldots a_{n}} (Q \sqcap Q')  & = &
 \begin{array}[t]{@{}l@{}}
 ((P_{1},\ldots, P_{n})\boxempty^{\Omega,a_{1}\ldots a_{n}} Q) \sqcap \mbox{}\\
 ((P_{1},\ldots, P_{n})\boxempty^{\Omega,a_{1}\ldots a_{n}}Q')
 \end{array}
  \nonumber\\
(P_{1},\ldots, P_{n})\boxempty^{\Omega,a_{1}\ldots a_{n}} (\bigbox_{j} b_{j}\boxaction Q_{j}) & = & 
  (\Omega \boxempty \bigbox_{i} a_{i}\boxaction P_{i}) \boxempty \bigbox_{j} b_{j}\boxaction Q_{j}
\nonumber\\
(P_{1},\ldots, P_{n})\boxempty^{\Omega,a_{1}\ldots a_{n}} (\Omega\boxempty\bigbox_{j} b_{j}\boxaction Q_{j}) & = & 
  (\Omega \boxempty \bigbox_{i} a_{i}\boxaction P_{i}) \boxempty \bigbox_{j} b_{j}\boxaction Q_{j}
\nonumber\\
(P_{1},\ldots, P_{n})\boxempty^{a_{1}\ldots a_{n}} (\Omega\boxempty\bigbox_{j} b_{j}\boxaction Q_{j}) & = & 
  (\Omega \boxempty \bigbox_{i} a_{i}\boxaction P_{i}) \boxempty \bigbox_{j} b_{j}\boxaction Q_{j}
\nonumber
\end{eqnarray}
For  $\mypar$ one adds the  auxiliary operator
$\mypar^{\Omega,a_{1}\ldots a_{n}}$ and the  equations:
\begin{eqnarray}
 (\Omega\boxempty\bigbox_{i} a_{i}\boxaction P_{i}) \mypar Q & = & (P_{1},\ldots, P_{n})\mypar^{\Omega,a_{1}\ldots a_{n}} Q\nonumber\\
 (P_{1},\ldots, P_{n})\mypar ^{\Omega,a_{1}\ldots a_{n}} (Q \sqcap Q')
 & = &
 \begin{array}[t]{@{}l@{}}
  ((P_{1},\ldots, P_{n})\mypar^{\Omega,a_{1}\ldots a_{n}} Q) \sqcap \mbox{}\\
  ((P_{1},\ldots, P_{n})\boxempty^{\Omega,a_{1}\ldots a_{n}} Q')\end{array}\nonumber\\
 (P_{1},\ldots, P_{n})\mypar^{\Omega,a_{1}\ldots a_{n}} (\bigbox_{j} b_{j}\boxaction Q_{j}) & = & 
 \Omega\boxempty  \bigbox_{a_{i}  =  b_{j}} a_{i}\boxaction(P_{i}\mypar Q_{j})
\nonumber\\
 (P_{1},\ldots, P_{n})\mypar^{\Omega,a_{1}\ldots a_{n}} (\Omega\boxempty\bigbox_{j} b_{j}\boxaction Q_{j}) & = & 
 \Omega\boxempty  \bigbox_{a_{i}  =  b_{j}} a_{i}\boxaction(P_{i}\mypar Q_{j})
\nonumber\\
 (P_{1},\ldots, P_{n})\mypar^{a_{1}\ldots a_{n}} (\Omega\boxempty\bigbox_{j} b_{j}\boxaction Q_{j}) & = & 
 \Omega\boxempty  \bigbox_{a_{i}  =  b_{j}} a_{i}\boxaction(P_{i}\mypar Q_{j})
\nonumber
\end{eqnarray}
Finally, for  $\myinterleave$ one simply adds  extra equations:
\vspace{-1ex}

\noindent
\begin{eqnarray}
(\Omega\boxempty\bigbox_{i = 1}^{n} a_{i}\boxaction P_{i}) \myinterleave^{l}  Q & = &  
\Omega\boxempty\bigbox_{i} a_{i}\boxaction ((P_{i} \myinterleave^{l}  Q) \boxempty (P_{i} \myinterleave^{r}  Q))\nonumber\\
Q \myinterleave^{r} (\Omega\boxempty\bigbox_{i = 1}^{n} a_{i}\boxaction P_{i}) & = & 
\Omega\boxempty\bigbox_{i} a_{i}\boxaction ((Q \myinterleave^{l}  P_{i}) \boxempty (Q \myinterleave^{r}  P_{i})) 
\nonumber
\end{eqnarray}

\section{Combining CSP and functional programming} \label{Functional}

To combine CSP with functional programming, specifically the
computational $\lambda$-calculus, we use the monad
$T_{\mathrm{CSP}(|,\Omega)}$ for the denotational semantics.  As
remarked above, CSP processes then become terms of type
$\prog{empty}$. However, as the constructors are polymorphic, it is
natural to go further and look for polymorphic versions of the
deconstructors. We therefore add polymorphic constructs
to $\lambda_{c}$ as follows:\vspace{1em}

{\bf  Constructors}
 
 \[\frac{M \type \sigma \quad N \type \sigma}{M \sqcap N \type \sigma}
 \quad\quad
 \frac{M\type\sigma}{a \rightarrow M\type\sigma}
 \quad\quad \Omega\type\sigma
\]

{\bf Unary Deconstructors}

\[\frac{M \type \sigma}{f(M) \type \sigma}
\quad\quad\quad
\frac{M \type \sigma}{M \backslash  a \type \sigma}\]
for any   relabelling function $f$, and any $a \in A$. (One should really restrict the allowable relabelling functions in order to keep the syntax finitary.)\vspace{1em}

{\bf Binary Deconstructors}

\[\frac{M \type \sigma \quad N \type \sigma}{M \boxempty N \type \sigma} 
\quad\quad\quad
\frac{M \type \sigma \quad N \type \tau}{M \mypar N \type \sigma \times \tau} 
\quad\quad\quad
\frac{M \type \sigma \quad N \type \tau}{M \myinterleave N \type \sigma \times \tau} \]
The idea of the two parallel constructs is to evaluate the two terms in parallel and then return the pair of the two values produced.
We did not include syntax for the two deterministic choice constructors as they are definable from $a \rightarrow - $ and $\Omega$  with the aid of the $\boxempty$ deconstructor.

For the denotational semantics, the semantics of types is given as usual using the monad $T_{\mathrm{CSP}(|,\Omega)}$, which we know exists by the general considerations of Section~\ref{Prelim}. These general considerations also yield a semantics for the constructors. For example, for every set $X$ we have the map:
\[\sqcap_{X} \type T_{\mathrm{CSP}(|,\Omega)}(X)^{2} \rightarrow T_{\mathrm{CSP}(|,\Omega)}(X)\]
which we can use for $X = \dsem{\sigma}$ to interpret terms $M \sqcap N\type \sigma$.

The homomorphic point of view also leads to an interpretation of the unary deconstructors, but using free algebras rather than just the initial one. For example, for relabelling by $f$
we need a function:
\[h_{\mathnormal{Rl}}\type T_{\mathrm{CSP}(|,\Omega)}(X) \rightarrow T_{\mathrm{CSP}(|,\Omega)}(X)\]
 We obtain this as the unique homomorphism extending the unit $\eta_{X}\! \type X\! \mathbin\rightarrow T_{\mathrm{CSP}(|,\Omega)}(X)$, equipping $T_{\mathrm{CSP}(|,\Omega)}(X)$ with the algebra structure $\mathcal{A} = (T_{\mathrm{CSP}(|,\Omega)}(X), \sqcap_{\mathcal{A}}\!, \bigbox_{{\mathcal{A}}}\!, \bigbox^{\Omega}_{{\mathcal{A}}}\!)$ where
 \[x \sqcap_{\mathcal{A}} y = x \sqcap_{X} y\]
for $x,y \in  T_{\mathrm{CSP}(|,\Omega)}(X)$, 
\begin{center} $(\bigbox_{\vec{a}})_{\mathcal{A}}(x_{1},\ldots,x_{n})  = (\bigbox_{f(\vec{a})})_{X}(x_{1},\ldots,x_{n}) $
\end{center}
and
\begin{center} $(\bigbox^{\Omega}_{\vec{a}})_{\mathcal{A}}(x_{1},\ldots,x_{n})  = (\bigbox^{\Omega}_{f(\vec{a})})_{X}(x_{1},\ldots,x_{n}) $
\end{center}

Concealment $- \backslash a$ can be treated analogously, but now following the treatment in the case of $\dfF$, and defining $\cal{A}$ by:
 \[x \sqcap_{\mathcal{A}} y = x \sqcap_{X} y\vspace{-1ex}\]
for $x,y \in  T_{\mathrm{CSP}(|,\Omega)}(X)$, 
\begin{center} $(\bigbox_{\vec{a}})_{\mathcal{A}}(x_{1},\ldots,x_{n})  = \left\{
\begin{array}{ll}
x_{j}  \sqcap (\Omega \boxempty \bigbox_{i \neq j} a_{i}\boxaction x_{i}) & (\mbox{if  } a = a_{j},\, \mbox{where} \,1 \leq j \leq n)\\
\bigbox_{i = 1}^{n} a_{i}\boxaction x_{i} & (\mbox{if  } a \neq \mbox{any} \; a_{i})
\end{array}
\right.$
\end{center}
and
\begin{center} $(\bigbox^{\Omega}_{\vec{a}})_{\mathcal{A}}(x_{1},\ldots,x_{n})  =   \left\{
\begin{array}{ll}
x_{j}  \sqcap (\Omega \boxempty \bigbox_{i \neq j} a_{i}\boxaction x_{i}) & (\mbox{if  } a = a_{j},\, \mbox{where} \,1 \leq j \leq n)\\
\Omega \boxempty \bigbox_{i = 1}^{n} a_{i}\boxaction x_{i} & (\mbox{if  } a \neq \mbox{any} \; a_{i})
\end{array}
\right.$
\end{center}
We here again make use of the deterministic choice operator made available by the presence of $\Omega$.

However, we cannot, of course, carry this on to binary deconstructors as we have no general algebraic treatment of them. We proceed instead by giving a concrete definition of them (and the other constructors and deconstructors). That is, we  give an explicit description of the free $\mathrm{CSP}(|,\Omega)$-algebra on a  set $X$ and define our operators in terms of that representation.

An  \emph{$X$-trace} is a pair $(w, x)$, where $w \in A^{*}$ and $x
\in X$; it is generally more suggestive to write $(w, x)$ as $wx$.
For any relabelling function $f$, we set $f(wx) = f(w)x$, 
and, for any $a \in A$, we set $wx\backslash a = (w\backslash a)x$.
An \emph{$X$-process} is a pair $(T,F)$ with $T$ a set of traces as well as
$X$-traces, and $F$ a set of failure pairs, satisfying the same five
conditions as in Section~\ref{divergence}, together with:
\begin{itemize}\parskip 0pt
\item[$2'$] $wx \in T \Rightarrow  w \in T$ (for $x \in X$)
\end{itemize}

The CSP operators are defined on $X$-processes exactly as before,
except that the two parallel operators now have more general types:
$$\mypar_{X,Y}, \myinterleave_{X,Y}\,\type
T_{\mathrm{CSP}(|,\Omega)}(X) \times
T_{\mathrm{CSP}(|,\Omega)}(Y) \rightarrow
T_{\mathrm{CSP}(|,\Omega)}(X\times Y)$$
We take $\futP(w) := \{a \in A \mid wa\in T_P\}$,  as before.

\[\begin{array}{@{}l@{~=~}l@{}}
        \Omega_\TX & \{\epsilon\} \\ \Omega_\RX & \emptyset
\\
	\mystop_{\mathnormal{\TX}}&\{\varepsilon\}
\\
	\mystop_{\mathnormal{\RX}}&\{(\varepsilon,W) \mid W \subseteq A\}
\\
	a\myactionsub_{\mathnormal{\TX}} P & \{\varepsilon\} \cup \{aw \mid w \in T_P\}
\\
	a\myactionsub_{\mathnormal{\RX}} P &\{(\varepsilon, W) \mid a \notin W\} \cup
	                                     \{(aw,W)\mid (w,W) \in F_P\}
\\
	 P \sqcap_{\mathnormal{\TX}} Q & T_P \cup T_Q
\\
	 P \sqcap_{\mathnormal{\RX}} Q & F_P \cup F_Q
\\
     P \boxempty_{\mathnormal{\TX}} Q & T_P \cup T_Q
\\
     P \boxempty_{\mathnormal{\RX}} Q &
           \{(\varepsilon,W)\mid (\varepsilon,W) \in F_P \cap F_Q\}
           \cup \mbox{}
           \{(w,W)\mid  w \neq \varepsilon,~ (w,W) \in F_P \cup F_Q\}
\\
f_\TX(P)  &  \{f(w) \mid w \in T_P\}
\\
f_\RX(P)  &  \{(f(w), W) \mid (w,f^{-1}(W) \cap \futP(w)) \in F_P\}
\\
P\backslash_\TX a  &  \{w \backslash a \mid w \in T_P\}
\\
P\backslash_\RX a  &  \{(w \backslash a,W) \mid (w,W \cup \{a\}) \in F_P\}
\\
P \mypar_\TXY Q & \{ w \mid w \in T_P \cap T_Q \cap A^*\} \cup
\{ w(x,y) \mid wx \mathbin\in T_P,\, wy \mathbin\in T_Q\}
\\
P \mypar_\RXY Q & \{ (w, W \cup V) \mid (w,W) \in F_P,~ (w,V) \in F_Q\}
\\
P \myinterleave_\TXY Q &
   \{ w \mid u \in T_P \cap A^*,~ v \mathbin\in T_Q \cap A^*,~ w \mathbin\in u\!\mid \!v\} \cup
\mbox{}\\
\multicolumn{2}{@{}r}{
   \{ w(x,y) \mid ux \in T_P,~ vy \mathbin\in T_Q,~ w \mathbin\in u\!\mid \!v\} }
\\
P \myinterleave_\RXY Q & \{ (w, W) \mid (u,W) \mathbin\in F_P,\, (v,W) \mathbin\in F_Q,\,
 w \mathbin\in u\!\mid \!v\}
 \end{array}\]
Here, much as before, we write $P \mathrel{\mathrm{op}_\FX} Q = (P \mathrel{\mathrm{op}_\TX} Q, P
 \mathrel{\mathrm{op}_\RX} Q)$ when defining the CSP operators on $X$-processes.
The $X$-processes also form the carrier of a
$\mathrm{CSP}(|,\Omega)$-algebra $\dFX$,
with the operators defined as follows:
\[\begin{array}{@{}l@{~=~}l@{}}
	 P \sqcap_{\mathnormal{\dTX}} Q & T_P \cup T_Q
\\
	 P \sqcap_{\mathnormal{\dRX}} Q & F_P \cup F_Q
\\
(\bigbox^\Omega_{\vec{a}})_{\mathnormal{\dTX}}(P_{1},\ldots,P_{n}) &
\{\varepsilon\} \cup \{a_{i}w \mid w \in T_{P_i}\}
\\
(\bigbox^\Omega_{\vec{a}})_{\mathnormal{\dRX}}(P_{1},\ldots,P_{n}) &
\{(a_{i}w,W) \mid (w,W) \in F_{P_i}\}
\\
(\bigbox_{\vec{a}})_{\mathnormal{\dTX}}(P_{1},\ldots,P_{n}) &
\{\varepsilon\} \cup \{a_{i}w \mid w \in T_{P_i}\}
\\
(\bigbox_{\vec{a}})_{\mathnormal{\dRX}}(P_{1},\ldots,P_{n}) &
  \{(\varepsilon, W) \mid W \cap \{a_{1},\dots,a_{n}\} = \emptyset\}
\cup \mbox{}
\\\multicolumn{2}{@{}r@{}}{
  \{(a_{i}w,W) \mid (w,W) \in F_{P_i}\}
}
\end{array}\]
The {\finitary} $X$-processes are those with a finite set of traces
and $X$-traces; they form the carrier of a
$\mathrm{CSP}(|,\Omega)$-algebra $ \mathnormal{\dfFX}$.

We now show that $\mathnormal{\dfFX}$ is the free $\mathrm{CSP}(|,\Omega)$-algebra over $X$. 
As is well known, the free algebra of a theory $\Th$ over a set $X$  is the same as the initial algebra of the theory $\Th^{+}$ obtained by  extending $\Th$ with constants $\underline{x}$ for each $x \in X$ but without changing the axioms. The unit map $\eta\type X \rightarrow T_{\Th}(X)$ sends $x \in X$ to the denotation of $\underline{x}$ in the initial algebra. 
We therefore show that $ \mathnormal{\dfFX}$,  extended to a  $\mathrm{CSP}(|,\Omega)^{+}$-algebra by taking
\[\dsem{\underline{x}} = (\{x\}, \emptyset)
\qquad\
(\mbox{for $x \in X$})\]
 is the initial $\mathrm{CSP}(|,\Omega)^{+}$-algebra.
We begin by looking at definability. 

\begin{lemma} \label{plus-definability}
The {\finitary} $X$-processes are those definable by closed
$\mathrm{CSP}(|,\Omega)^{+}$ terms.
\end{lemma}
\begin{proof}
The proof goes just as the one for Lemma~\ref{Definability_Omega},
using that Lemma~\ref{Refusal_accounting_Omega} applies just as well
to finitary $X$-processes, but this time we have
\\[1ex]\mbox{}\hfill$\displaystyle
  P = \bigsqcap_{i}\bigbox_{a \in V_{i}} a\myaction P_{a}
  ~\sqcap~ \left(\Omega \boxempty \bigbox_{a \in T_P} a\myaction
  P_{a}\right)
  ~\sqcap~ \bigsqcap_{x \in T_P} \underline{x}
$\hfill\mbox{}\hspace{-2em}\vspace{1ex}
\qed 
 \end{proof}
   
Next, we say that a closed $\mathrm{CSP}(|,\Omega)^{+}$-term $t$ is
in \emph{normal form} if it is has one of the following two forms:
\vspace{-1ex}
\[ \bigsqcap_{L \in \mathcal{L}}\bigbox_{a \in L} a\boxaction t_{a}
 \sqcap \bigsqcap_{x \in J} \underline{x}
\qquad \mbox{or} \qquad \left(\Omega \boxempty \bigbox_{a \in K} a\boxaction t_{a}\right)
 \sqcap \bigsqcap_{x \in J} \underline{x}\]
where, as appropriate,  $\mathcal{L}$ is a finite non-empty saturated collection of finite sets of
actions, $J\subseteq_{\mathrm{fin}} X$, $K \subseteq_{\mathrm{fin}} A$, and each term $t_{a}$ is in normal form. 

\begin{lemma} \label{equal normal forms}
Two normal forms are identical  if they have the same denotation in $\dfFX$.
\end{lemma}
\begin{proof}
Consider two normal forms with the same denotation in $\dfFX$, say $(T,F)$.
As $(\varepsilon,\emptyset) \in F$ iff $F$ is the denotation of a
normal form of the first form (rather than the second), both normal
forms must be of the same form.
Thus, there are two cases to consider, the first of which concerns two
forms:
\[ \bigsqcap_{L \in \mathcal{L}}\bigbox_{a \in L} a\boxaction t_{a}
\sqcap \bigsqcap_{x \in J} \underline{x}
\quad\quad\quad
 \bigsqcap_{L' \in \mathcal{L'}}\bigbox_{a' \in L'} a'\boxaction
 t'_{a'} \sqcap \bigsqcap_{x \in J'} \underline{x}
\]
We argue by induction on the sum of the sizes of the two normal forms.
We evidently have that $J=J'$.
Next, if $a \in \bigcup \mathcal{L}$ then $a \in T$ and so $a \in
\bigcup \mathcal{L'}$; we therefore have that    $\bigcup \mathcal{L}
\subseteq \bigcup \mathcal{L'}$. Now, if $L \in \mathcal{L}$  then 
$(\varepsilon, (\bigcup \mathcal{L'})\backslash L) \in F$; so for some $L' \in \mathcal{L}$ we have $L' \cap ((\bigcup\mathcal{L'}) \backslash L) = \emptyset$, and so $L' \subseteq L$. As $\mathcal{L'}$ is saturated, it follows by the previous remark that $L \in \mathcal{L}'$. So we have the inclusion $\mathcal{L}\subseteq \mathcal{L'}$ and then, arguing symmetrically, equality.

Finally, the denotations of $t_{a}$ and $t'_{a}$, for $a \in \bigcup
\mathcal{L} = \bigcup \mathcal{L'}$ are the same, as they are
determined by $T$ and $F$, being $\{w \mid aw \in T\}$ and
$\{(w,W)\mid (aw,W) \in F\}$, and the argument concludes, using the inductive hypothesis.

The other case concerns normal forms:
\[ \left(\Omega \boxempty \bigbox_{a \in K} a\boxaction t_{a}\right)
 \sqcap \bigsqcap_{x \in J} \underline{x}
\quad\quad\quad
 \left(\Omega \boxempty \bigbox_{a' \in K'} a'\boxaction t'_{a}\right)
 \sqcap \bigsqcap_{x \in J'} \underline{x}\]
Much as  before we find $J=J'$, $K=K'$, and $t_a = t_a$ for $a \in K$.
\qed \end{proof}

\begin{lemma} \label{plus-completeness}
$\mathrm{CSP}(|,\Omega)^{+}$ is ground complete with respect to $\dfFX$.
\end{lemma}

\begin{proof}
As before, a straightforward induction shows that every term has a
normal form, and then completeness follows by Lemma~\ref{equal normal forms}.
\qed \end{proof}

\begin{theorem} \label{Free} The algebra $\dfFX$ is the free $\mathrm{CSP}(|,\Omega)$-algebra over $X$. 
\end{theorem}
\begin{proof} It follows from Lemmas~\ref{plus-definability} and
  ~\ref{plus-completeness} that $\dfFX^{+}$ is the
  initial $\mathrm{CSP}(|,\Omega)^{+}$-algebra.
\qed \end{proof}

As with any finitary equational theory, $\mathrm{CSP}(|,\Omega)$ is
equationally complete with respect to   $\dfFX$  when
$X$  is infinite. It  is not difficult to go a little further and show
that this also holds when $X$ is only required to be non-empty, and,
even, if $A$ is infinite, when it is empty.

Now that we have an explicit representation of the free
$\mathrm{CSP}(|,\Omega)$-monad in terms of $X$-processes, we
indicate how to use it to give the semantics of the computational
$\lambda$-calculus. First we need the structure of the monad. As we know from the above, the unit
$\eta_{X}\type X  \rightarrow  T_{\mathrm{CSP}(|,\Omega)}(X)$ is the map $x \mapsto (\{x\}, \emptyset)$.
Next, we need the homomorphic extension $g^{\dagger}\type
\mathnormal{\dfF}(X) \rightarrow  \mathnormal{\dfF}(Y)$ of a given map
\mbox{$g \type X \rightarrow  \mathnormal{\dfF}(Y)$}, i.e., the unique such homomorphism making the following diagram commute:
\begin{diagram}
   X &  &\\
  \dTo^{\eta_{X}} & \rdTo^{g} & & \\
  T_{\mathrm{CSP}(|,\Omega)}(X) &  \rTo^{g^{\dagger}} &  T_{\mathrm{CSP}(|,\Omega)}(Y)
\end{diagram}
This is given by:
\[(g^{\dagger}(P))_{\T}= \{v \mid v \in T_P \cap A^*\} \cup \{vw \mid vx \in T_P,~ w \in g(x)_{\T} \}\]
\vspace{-1.5em}
\[(g^{\dagger}(P))_{\R} = \{(v,V) \in F_P\} \cup \{(vw,W) \mid vx \in T_P,~ (w,W) \in g(x)_{\R} \}\]

As regards the constructors and deconstructors, we have already given explicit representations of them as functions over (finitary)
$X$-processes. We have also already given homomorphic treatments of the unary deconstructors. We finally give treatments of the binary deconstructors as unique solutions to equations, along similar lines to their treatment in the case of  $\dfF$.
Observe that:
\begin{center}
	$(\bigbox_{\vec{a}})_X(P_1,\ldots,P_n) = a_1P_1 \boxempty_X a_2P_2 \boxempty_X  \ldots \boxempty_X a_{n}P_{n}$
\\
	$(\bigbox^\Omega_{\vec{a}})_X(P_1,\ldots,P_n) = \Omega
\boxempty_X a_1P_1 \boxempty_X a_2P_2 \boxempty_X  \ldots \boxempty_X a_{n}P_{n}$
\end{center}
Using this, one finds that $\boxempty_X$, $\boxempty_X^{\Omega,a_1\ldots a_n}$ and $\boxempty_X^{a_1\ldots a_n}$, the latter defined as in equation~(\ref{box-super-def}), are the unique functions which satisfy the evident analogues of equations~(\ref{box-def}) together with, making another use of the form of external choice made available by $\Omega$:
\[\eta(x) \boxempty P = \eta(x) \sqcap_X (\Omega \boxempty P)
\]
and
\begin{center}
	$(P_1,\ldots,P_n)\boxempty^{a_1\ldots a_n} \eta(x) = (\bigbox^\Omega_{\vec{a}})_{X}(P_1,\ldots,P_n) \sqcap_X \eta(x)$
\\
	$(P_1,\ldots,P_n)\boxempty^{\Omega,a_1\ldots a_n} \eta(x) = (\bigbox^\Omega_{\vec{a}})_{X}(P_1,\ldots,P_n) \sqcap_X \eta(x)$\\
\end{center}
As regards concurrency, we define 
$$\mypar_{X,Y} \type T_{\mathrm{CSP}(|,\Omega)}(X) \times T_{\mathrm{CSP}(|,\Omega)}(Y) \rightarrow T_{\mathrm{CSP}(|,\Omega)}(X \times Y)$$ together with functions 
\[\mypar^{a_1\ldots a_n}_{X,Y}\type T_{\mathrm{CSP}(|,\Omega)}(X)^n
\times T_{\mathrm{CSP}(|,\Omega)}(Y)\rightarrow
T_{\mathrm{CSP}(|,\Omega)}(X\times Y)\]
\vspace{-1.5em}
\[\mypar^{\Omega,a_1\ldots a_n}_{X,Y}\type T_{\mathrm{CSP}(|,\Omega)}(X)^n \times T_{\mathrm{CSP}(|,\Omega)}(Y)\rightarrow T_{\mathrm{CSP}(|,\Omega)}(X\times Y)\]
\vspace{-1em}
\[\mypar^x_{X,Y}\type T_{\mathrm{CSP}(|,\Omega)}(Y) \rightarrow T_{\mathrm{CSP}(|,\Omega)}(X \times Y)\]
 where the $a_i \in A$ are all different, and $x \in X$, by the analogues of equations~(\ref{par-def}) above, together with:
\[\begin{array}{lcl}
\eta(x) \mypar Q & = & \mypar^{x} (Q)\\
&&\\
\mypar^{x}(P \sqcap Q) & = & \mypar^x(P)\, \sqcap\, \mypar^x(Q)\\
\mypar^{x}(\bigbox_{i=1}^n a_i\boxaction P_i) & = & \Omega\\
\mypar^{x}(\Omega\boxempty\bigbox_{i=1}^n a_i\boxaction P_i) & = & \Omega\\
\mypar^{x}(\eta(y)) & = & \eta((x,y))\\
&&\\
(P_1,\ldots,P_n)\mypar^{a_1\ldots a_n} \eta(x) & = & \Omega
\\
(P_1,\ldots,P_n)\mypar^{\Omega,a_1\ldots a_n} \eta(x) & = & \Omega
\end{array}\]
Much as before, the equations have a unique solution, with the
$\mypar$ component being $\mypar_{X,Y}$.

As regards interleaving, we define 
$$\myinterleave^{l}_{X,Y}, \myinterleave^{r}_{X,Y}\type T_{\mathrm{CSP}(|,\Omega)}(X) \times T_{\mathrm{CSP}(|,\Omega)}(Y) \rightarrow T_{\mathrm{CSP}(|,\Omega)}(X \times Y)$$ by:
\[\begin{array}{lcl}
P \myinterleave^l_\dfTXY Q & = & \{\varepsilon\} \cup
   \{ w \mid u \in T_P \cap A^*,~ v \mathbin\in T_Q \cap A^*,~ w \mathbin\in u\!\mid^{l} \!v\} \cup
\mbox{} \\ &&
   \{ w(x,y) \mid ux \in T_P,~ vy \mathbin\in T_Q,~ w \mathbin\in u \mid^{l}v \vee (u = v = w = \varepsilon)\}\\
  &&\\
 P\myinterleave^{l}_\dfRXY Q & = &  \{(\varepsilon, W) \mid
   (\varepsilon, W) \in F_P\} \cup \mbox{}
\\&&
   \{(w,W)\mid (u,W) \in F_P,~ (v,W) \in F_Q,~  w \in u \mid^{l}v\}\\
  &&\\
 P\myinterleave^{r}_{X,Y} Q & = & Q\myinterleave^{l}_{Y,X} P 
 \end{array}\]
One has that:
\[P\myinterleave_{X,Y} Q =  P\myinterleave^{l}_{X,Y} Q \,\boxempty\,
P \myinterleave^{r}_{X,Y} Q\]
 and that $\myinterleave^{l}_{X,Y}, \myinterleave^{r}_{X,Y}$ are components of the unique solutions to the analogues of equations~(\ref{inter-def}) above, together with:
\[\begin{array}{lcl}
\eta(x) \myinterleave^{l} Q & = & \myinterleave^{l,x}(Q)\\
&&\\
 \myinterleave^{l,x}(P \sqcap Q) & = & \myinterleave^{l,x}(P)  \, \sqcap \, \myinterleave^{l,x}(Q)\\
 \myinterleave^{l,x}(\bigbox_{i=1}^n a_i \boxaction P_i) & = & \Omega\\
 \myinterleave^{l,x}(\Omega\boxempty\bigbox_{i=1}^n a_i \boxaction P_i) & = & \Omega\\
 \myinterleave^{l,x}(\eta(y)) & = & \eta(x,y) 
\end{array}\]
and corresponding equations for $\myinterleave^{r}$ and $\myinterleave^{r,y} $.

It would be interesting to check more completely which of the usual laws, as found in, e.g.,~\cite{BHR84,Hoa85,DeN85}, the CSP operators at the level of free $\mathrm{CSP}(|,\Omega)$-algebras obey.  Note that some adjustments  need to be made due to varying types. For example, $\mypar$ is commutative, which here means that the following equation holds:
\[T_{\mathrm{CSP}(|,\Omega)}(\gamma_{X,Y}) (P \mypar_{X,Y} Q) = Q \mypar_{Y, X} P\]
where $\gamma\type  X \times Y  \rightarrow Y \times X$ is the commutativity map $(x,y) \mapsto (y,x)$.

\subsection{Termination}\label{termination}

As remarked in the introduction, termination and sequencing are available in a standard way for terms of type $\prog{unit}$. Syntactically, we regard $\prog{skip}$ as an abbreviation for $\ast$ and $M;N$ as  one for $(\lambda x\type \prog{unit}. N)(M)$ where $x$ does not occur free in $N$; semantically, we have a corresponding element of, and binary operator over, the free ${\mathrm{CSP}(|,\Omega)}$-algebra  on the one-point set.  

Let us use these ideas to treat CSP extended with termination and sequencing. We work with the finitary $\{\checkmark\}$-processes representation of $T_{\mathrm{CSP}(|,\Omega)}(\{\checkmark\})$. 
Then, following the above prescription, termination and sequencing are given by:
\[\prog{SKIP} =  \{\checkmark\} \quad\quad\quad\quad P;Q = (x \in \{\checkmark\} \mapsto Q)^{\dagger}(P)\] 
For general reasons, termination and sequencing, so-defined, form a monoid and sequencing commutes with all constructors in its first argument. For example we have that:
\[\bigbox_{i = 1}^{n} a_{i}\boxaction (P_{i};Q) = (\bigbox_{i = 1}^{n} a_{i}\boxaction P_{i});Q\]
Composition further commutes with $\sqcap$ in its second argument.
 
The deconstructors are defined as above except that in the case of the concurrency operators one has to adjust  $\mypar_{\{\checkmark\},\{\checkmark\}}$ and $\myinterleave_{\{\checkmark\},\{\checkmark\}} $ so that they remain within the world of the $\{\checkmark\}$-processes; this can be done by postcomposing them with the evident bijection between $\{\checkmark\} \times \{\checkmark\}$-processes and 
$\{\checkmark\}$-processes, and all this restricts to the finitary processes. Alternatively one can directly consider these adjusted operators as deconstructors over the (finitary)
$\{\checkmark\}$-processes.

The $\{\checkmark\}$-processes are essentially the elements of the
stable failures model of~\cite{Ros98}.
More precisely, one can define a bijection from Roscoe's model to  our $\{\checkmark\}$-processes 
by setting $\theta(T,F) = (T,F')$ where
$$F' = \{(w,W) \in A^*\times{\cal P}(A) \mid (w,W\cup\{\checkmark\})\in F\}$$
The inverse of $\theta$ sends $F'$ to the set:
$$\{(w,W),(w,W\cup\{\checkmark\}) \mid (w,W)\in F'\}
          \cup \mbox{}$$
 \vspace{-2em}
$$\{(w,W) \mid w\checkmark \in T \wedge W\subseteq A\}
\cup \{(w\checkmark,W) \mid w\checkmark \in T \wedge W \in A \cup \{\checkmark\}\}$$
and is a homomorphism between  all our operators,
whether constructors, deconstructors, termination, or sequencing (suitably defined),
and the corresponding ones defined for Roscoe's model.

\section{Discussion} \label{the-end}

We have shown the possibility of  a principled combination of CSP and functional programming from the viewpoint of the algebraic theory of effects. The main missing ingredient is an algebraic treatment of binary deconstructors, although we were able to partially circumvent that by giving explicit definitions of them. Also missing are a logic for proving properties of these deconstructors, an operational semantics, and a treatment that includes recursion. 

As regards a logic, it may prove possible to adapt the logical ideas of~\cite{PPr08,PPr09} to handle binary deconstructors; the main proof principle would then be that of {\em computation induction}, that if a proposition holds for all `values' (i.e., elements of a given set $X$) and if it holds for the applications of each constructor to any given `computations' (i.e., elements of $T(X)$) for which it is assumed to hold, then it holds for all computations.  
We do not anticipate any difficulty in giving an operational semantics for the above combination of the computational $\lambda$-calculus and CSP and proving an adequacy theorem.

To treat recursion algebraically, one passes from equational theories
to inequational theories $\Th$ (inequations have the form $t \leq u$,
for terms $t$, $u$ in a given signature $\Sigma$); inequational theories can include equations, regarding an equation  as  two evident  inequations.  There is a natural inequational logic for deducing consequences of the axioms: one simply drops symmetry from the logic for equations~\cite{Blo76}.
Then $\Sigma$-algebras
and $\Th$-algebras are taken in the category of $\omega$-cpos and
continuous functions, a free algebra monad always exists, just as
in the case of sets, and the logic is complete for the class of such algebras.
 One includes  a divergence constant $\Omega$ in the signature and the axiom
\[\Omega \leq x\] 
so that $\Th$-algebras always have a least element. Recursive
definitions are then modelled by {least} fixed-points in the usual way.
See~\cite{HPP06,Plo06} for some further explanations.

The three classical powerdomains: convex (aka Plotkin), lower (aka Hoare) and upper (aka Smyth) provide a useful illustration of these ideas~\cite{GHK03,HPP06}. One takes as  signature a binary operation symbol $\sqcap$, to retain notational consistency with the present paper (a more neutral symbol, such as $\cup$, is normally used instead), and the constant $\Omega$; one takes the theory to be that $\sqcap$ is a semilattice (meaning, as before, that associativity, commutativity and idempotence hold) and that, as given above, $\Omega$ is the least element with respect to the ordering $\leq$. This gives an algebraic account of the convex powerdomain. 

If one adds that $\Omega$ is the zero of the semilattice (which is equivalent, in the present context, to the inequation $x \leq x \sqcap y$) one obtains instead an algebraic account of the lower powerdomain. One then further has the notationally counterintuitive facts that $x \leq y$ is equivalent to $y \sqsubseteq x$, with $\sqsubseteq$ defined as in Section~\ref{CSPconstructors1}, and that $x \sqcap y$ is the supremum of $x$ and $y$ with respect to $\leq$; in models, $\leq$ typically corresponds to subset. It would be more natural in this case to use the dual order to $\sqsubseteq$ and to write $\sqcup$ instead of $\sqcap$, when we would be dealing with a join-semilattice with a least element whose order coincides with $\leq$.

If one adds instead that $x \sqcap y \leq x$, one obtains an algebraic account of the upper powerdomain. One now has that $x \leq y$ is equivalent in this context to $x \sqsubseteq y$, that $x \sqcap y$ is the greatest lower bound of $x$ and $y$, and that $x \sqcap \Omega = \Omega$ (but this latter fact is not equivalent in inequational logic to $x \sqcap y \leq x$); in models, $\leq$ typically corresponds to superset.
The notations $\sqcap$ and $\sqsubseteq$ are therefore more intuitive in the upper case, and there one has a meet-semilattice with a least element whose order coincides with $\leq$. 

It will be clear from these considerations that the stable failures model fits into the pattern of the lower powerdomain and that the failures/divergences model fits into the pattern of the upper powerdomain. In the case of the stable failures model it is natural, in the light of the above considerations, to take  $\Th$ to be ${\mathrm{CSP}(|,\Omega)}$ together with the axiom $\Omega \leq x$. The $X$-processes with countably many traces presumably form the free algebra over $X$, considered as a discrete $\omega$-cpo; one should also characterise more general cases than discrete $\omega$-cpos. 

One should also investigate whether a fragment of the failures/divergences model forms the initial model of an appropriate theory, and look at the free models of such a theory. The theory might well be found by analogy with our work on the stable failures model, substituting (\ref{chaos}) for (\ref{Omega}) and, perhaps, using the mixed-choice constructor, defined below, to overcome any difficulties with the deconstructors. One would expect the initial model to contain only finitely-generable processes, meaning those which, at any trace, either branch finitely or diverge (and see the discussion in~\cite{Ros98}).

Our initial division of our selection of CSP operators into constructors and deconstructors was natural, although it turned out that a somewhat different division, with `restricted' constructors, resulted in what seemed to be a better analysis (we were not able to rule out the possibility that there are alternative, indirect, definitions of the deconstructors with the original choice of constructors).  One of these restricted constructors was a deterministic choice operator making use of the divergence constant $\Omega$. There should surely, however, also be a development without divergence that allows the interpretation of the combination of CSP and functional programming. 

We were, however, not able to do this using $\mathrm{CSP}(|)$:
the free algebra does not seem to support a suitable definition of concealment, whether defined directly or via a homomorphism. For example
a straightforward extension of the homomorphic treatment of concealment in the case of the initial algebra
(cf.\ Section~\ref{CSPconstructors2}) would give
$$(a.\underline{x} \boxempty b.\mystop) \backslash a = \underline{x} \sqcap (\underline{x} \boxempty b.\mystop)$$
However, our approach requires the right-hand side to be equivalent to
a term built from constructors only, but no natural candidates came
forward---all choices that came to mind lead to unwanted identifications.

We conjecture that, taking instead, as constructor, a {\em mixed-choice} operator of the form:
\[\bigbox_{i}\alpha_{i}.x_{i}\]
where each $\alpha_{i}$ is either an action or $\tau$, would lead to a satisfactory theory. This new operator is given by the equation:
\[\bigbox_{i}\alpha_{i}.x_{i} = \bigsqcap_{\alpha_{i} = \tau}x_{i} \;\sqcap 
\left(\bigbox_{\alpha_{i} = \tau}x_{i}\; \boxempty 
\bigbox_{\alpha_{i} \neq \tau}\alpha_{i}.x_{i} \right)\] 
and there is a homomorphic relationship with concealment:
\[(\bigbox_{i}\alpha_{i}.x_{i})\backslash a = \bigbox_{i}(\alpha_{i}\backslash a).(x_{i}\backslash a)\]
(with the evident understanding of $\alpha_{i}\backslash a$). Note that in the stable failures model we have the equation:
\[\bigbox_{i}\alpha_{i}.x_{i} = \bigsqcap_{\alpha_{i} = \tau}x_{i} \;\sqcap 
\left(\Omega \boxempty 
\bigbox_{\alpha_{i} \neq \tau}\alpha_{i}.x_{i} \right)\] 
which is presumably why the deterministic choice operator available in the presence of $\Omega$ played so central a r\^{o}le there.

In a different direction, one might  also ask if there is some problem if we alternatively take an extended set of operators as constructors. For example, why not add relabelling with its equations to the axioms? As the axioms inductively determine relabelling on the {\finitary} refusal sets model, that would still be the initial algebra, and the same holds if we add any of the other operators we have taken as deconstructors. 

However, the $X$-refusal sets would not longer be the free algebra, as there would be extra elements, such as $f(x)$ for $x \in X$, where $f$ is a relabelling function.
We would also get some undesired equations holding between terms of the computational $\lambda$-calculus. For any $n$-ary constructor $\mathrm{op}$ and evaluation context $E[-]$, one has in the monadic semantics: 
\[E[\mathrm{op}(M_1,\ldots,M_n)] = \mathrm{op}(E[M_{1}],\ldots, E[M_{n}])\]
So one would have
$E[f(M)] = f(E[M])$ if one took relabelling as a constructor, and, as another example,
one would have  $E[M \mypar N] = E[M] \mypar E[N]$ if one took the concurrency operator as a constructor.
   
It will be clear to the reader that, in principle, one can investigate other process calculi and their combination with functional programming in a similar way. For example for Milner's CCS~\cite{Mil80} one could take action prefix (with names, conames and $\tau$) together with $\prog{NIL}$ and the sum operator as constructors, and as axioms that we have a semilattice with a zero, for strong bisimulation,  together with the usual $\tau$-laws, if we additionally wish to consider weak bisimulation. The deconstructors would be renaming, hiding, and parallel, and all should have suitable polymorphic versions in the functional programming context. Other process calculi such as the $\pi$-calculus~\cite{DW03,Sta08}, or even the stochastic $\pi$-calculus~\cite{Pri95,KS08}, 
might be dealt with similarly. In much the same way, one could combine parallelism with a global store with functional programming, following the algebraic account of the resumptions monad~\cite{HPP06,AP09} where the constructors are  the two standard ones for global store~\cite{PP02}, a nondeterministic choice operation, and a unary `suspension'  operation.

A well-known feature of the monadic approach~\cite{HPP06} is that it is often possible to combine different effects in a modular way. For example, the global side-effects monad is $(S \times -)^{S}$ where $S$ is a suitable set of states. A common combination of it with another monad $T$ is the monad $T(S \times -)^{S}$. So, taking $T = T_{\mathrm{CSP}(\mid)}$, for example,  we get a combination of CSP with global side-effects.  

As another example, given a monoid $M$, one has the $M$-action monad $M \times -$ which supports a unary $M$-action effect constructor $m.-$, parameterised by elements $m$ of the monoid. One might use this monad to model the passage of time, taking $M$ to be, for example, the monoid of the natural numbers $\bbbn$ under addition. A suitable combination of this monad with ones for CSP may yield  helpful analyses of timed CSP~\cite{RR99,OS06}, with $\mathit{Wait}\;  n; - $ given by the $\bbbn$-action effect constructor. 
We therefore have a very rich space of possible combinations of process calculi, functional programming and other effects, and we hope that some of these  prove useful.

Finally, we note that there is no general account of how the equations
used in the algebraic theory of effects arise.
In such cases as global state, nondeterminism or probability, there are natural axioms and monads already available, and it is encouraging that the two are equivalent~\cite{PP02,HPP06}.
One could investigate using  operational methods and behavioural equivalences to determine the equations, and it would be interesting to do so. Another approach is the use of `test algebras'~\cite{SS06,KP09}. In the case of process calculi one naturally uses operational methods; however the resulting axioms may not be very modular, or very natural mathematically, and, all in all, in this respect the situation is not satisfactory.


\bibliographystyle{plain}


\section*{Appendix: The computational $\lambda$-calculus}\label{comp}

In this appendix, we sketch (a slight variant of) the syntax and semantics of Moggi's computational $\lambda$-calculus, or
$\lambda_c$-calculus~\cite{Mog89,Mog91}.
It has types given by:
\[
{  \sigma ::= \prog{b} \mid  \prog{unit}\mid  \sigma \times \sigma \mid \prog{empty}  \mid \sigma 
\rightarrow \sigma }
\]
where $\prog{b}$ ranges over a given set of { base types}, e.g., $\prog{nat}$;
the type construction ${  T\sigma }$ may be
defined to be ${\prog{unit}  \rightarrow \sigma }$.
The terms of the $\lambda_c$-calculus are given by:
\[
M ::=  x\mid g(M)  \mid \ast \mid \prog{in} M \mid (M,M) \mid\prog{fst} M\mid \prog{snd} M \mid \lambda x\type \sigma. M \mid MM  
\]
where $g$ ranges over given unary function symbols of  given types $\sigma \rightarrow \tau$, such as $0 \type \prog{unit} \rightarrow \prog{nat}$ or
$\prog{succ} \type \prog{nat} \rightarrow \prog{nat}$, if we want the natural numbers, or $\mathrm{op}\type T(\sigma) \times \ldots \times T(\sigma) \rightarrow T(\sigma)$ for some operation symbol from a theory for which $T$ is the free algebra monad. There are standard notions of free and bound variables and of closed terms and substitution;
there are also standard typing rules for judgements $\Gamma \vdash M \type \sigma$, that the term 
$M$ has type $\sigma$ in the context $\Gamma$ (contexts have the form 
$\Gamma = x_1:\sigma_1, \ldots, x_n:\sigma_n$), including:
\[\frac{\Gamma \vdash M \type \prog{empty}}{\Gamma \vdash \prog{in} M \type \sigma}\]

A $\lambda_c$-model (on the category of sets---Moggi worked more generally) consists of a monad $T$, together with enough information to interpret basic types and the given function symbols.
So there is a given set $\dsem{b}$ to interpret each basic type $b$,
and then every type $\sigma$ receives an interpretation as a set
$\dsem{\sigma}$; for example $\dsem{\prog{empty}} = \emptyset$.
There is also given a map 
$\dsem{\sigma}\rightarrow T(\dsem{\tau})$ to interpret every given unary function symbol $g\type \sigma \rightarrow \tau$. 
A term $\Gamma \vdash M \type \sigma $ of type ${\sigma}$ in context $\Gamma $ is modelled
by a map $\dsem{M} \type \lsem{  \Gamma }\rsem \rightarrow T\lsem{  \sigma}\rsem$
(where  
${\dsem{x_1\type \sigma_1, \ldots, x_n\type \sigma_n} = \dsem{\sigma_1} \times \ldots \times  \dsem{\sigma_n}}$).
For example, if  $\Gamma \vdash \prog{in} M \type \sigma $ then $\dsem{\prog{in} M} = \comp{0_{\dsem{\sigma}}}{\dsem{M}}$
(where, for any set $X$, $0_{X}$ is the unique map from $\emptyset$ to $X$).

We define values and evaluation contexts. Values can be thought of  as (syntax for) completed computations, and are defined by:
\[ V ::=  x \mid \ast \mid (V,V) \mid \prog{in} V \mid \lambda x\type \sigma. M 
\]
together with clauses such as:
\[V ::=  0 \mid \prog{succ} (V)\]
depending on the choice of basic types and given function symbols.
We may then  define evaluation contexts by:
\[
E ::= [-] \mid \prog{in} E \mid (E,M) \mid (V,E) \mid EM \mid VE \mid \prog{fst}(E) \mid \prog{snd}(E)
\]
together with clauses such as:
\[E ::=  \prog{succ} (E)\]
depending on the choice of basic types and given function symbols.
We write $E[M]$ for the term obtained by replacing the `hole' $[-]$ in an evaluation term $E$ by a term $M$.
The computational thought behind evaluation 
contexts is that in a program of the form ${  E[M]}$, 
the first computational step 
arises within $M$.  
\end{document}